\documentclass[11pt]{article}
\usepackage[a4paper,margin=1in]{geometry}
\usepackage{amsmath,amssymb,amsthm,amsfonts,bm,bbm,mathtools}
\usepackage{booktabs}
\usepackage{multirow}
\usepackage{array}
\usepackage{enumitem}
\usepackage{mathrsfs}
\usepackage{fullpage}
\usepackage{hyperref}
\hypersetup{hidelinks}
\usepackage{setspace}
\usepackage{natbib}
\newtheorem{theorem}{Theorem}[section]
\newtheorem{lemma}{Lemma}[section]
\newtheorem{corollary}{Corollary}[section]
\newtheorem{proposition}{Proposition}[section]
\newtheorem{remark}{Remark}[section]
\newtheorem{assumption}{Assumption}[section]

\newcommand{\R}{\mathbb{R}}
\newcommand{\E}{\mathbb{E}}
\newcommand{\var}{\mathrm{Var}}
\newcommand{\Cov}{\mathrm{Cov}}
\newcommand{\Corr}{\mathrm{Corr}}
\newcommand{\Prb}{\mathbb{P}}
\newcommand{\ind}{\mathbbm{1}}
\newcommand{\T}{\mathsf{T}}
\newcommand{\tr}{\mathrm{tr}}
\newcommand{\diag}{\mathrm{diag}}

\newcommand{\norm}[1]{\left\lVert #1\right\rVert}
\newcommand{\abs}[1]{\left\lvert #1\right\rvert}

\newcommand{\dto}{\stackrel{d}{\to}}
\newcommand{\pto}{\stackrel{p}{\to}}

\newcommand{\cir}{\circ}

\newcommand{\de}{\delta}

\def\mR{\mathbb{R}}

\def\A{\mathbf{A}}
\def\B{\mathbf{B}}

\def\D{\mathbf{D}}
\def\F{\mathbf{F}}

\def\I{\mathbf{I}}

\def\Q{\mathbf{Q}}
\def\R{\mathbf{R}}

\def\W{\mathbf{W}}
\def\X{\mathbf{X}}
\def\Y{\mathbf{Y}}

\def\b{{\bm b}}

\def\f{{\bm f}}
\def\g{{\bm g}}

\def\s{{\bm s}}
\def\u{{\bm u}}

\def\x{{\bm x}}
\def\y{{\bm y}}
\def\z{{\bm z}}

\def\bmS{{\bm \Sigma}}

\def\bD{{\bm \Delta}}

\def\bmb{{\bm \beta}}

\def\bml{{\bm \lambda}}

\def\bme{{\bm \varepsilon}}
\def\bal{{\bm \alpha}}

\def\bet{{\bm \eta}}
\def\bzero{{\bm 0}}
\def\bone{{\bm 1}}

\allowdisplaybreaks[4]

\title{\bfseries High dimensional alpha test for linear factor pricing model with $L_q$-norm}
\author{Ping Zhao$^1$, Huifang Ma$^2$ and Long Feng$^2$\\
$^1$School of Mathematical Scicence, Tiangong University\\$^2$School of Statistics and Data Science, LEBPS, KLMDASR,\\
 AAIS and LPMC, Nankai University}
\date{}

\begin{document}
\maketitle

\begin{abstract}
We consider testing zero pricing errors in high-dimensional linear factor pricing models. Existing methods are mainly based on either an $L_2$ statistic, which is effective under dense alternatives, or an $L_\infty$ statistic, which is powerful under very sparse alternatives. To bridge these two regimes, we develop a class of $L_q$-based tests for finite $q$, including the practically useful $L_4$ and $L_6$ cases. We show that larger $q$ leads to greater sensitivity to sparse alternatives. We further establish the asymptotic independence between the $L_\infty$ statistic and the $L_q$ statistic for any finite $q$, which motivates a Cauchy combination test that adapts to a broad range of sparsity levels. Simulation studies and a real-data analysis show that the proposed methods are more robust to the unknown sparsity of the alternative and can outperform existing procedures in finite samples.
\end{abstract}


\section{Introduction}
We consider testing for the presence of alpha in linear factor pricing models with observed factors when the number of securities, \(N\), is large relative to the time dimension, \(T\). This setup covers prominent examples such as the capital asset pricing model and the arbitrage pricing theory; see \citet{Sharpe1964,Lintner1965,Ross1976}.
Let
\[
Y_{it} = \alpha_i + \bmb_i^{\T} \f_t + \varepsilon_{it},
\qquad i=1,\ldots,N,\quad t=1,\ldots,T,
\]
denote a linear factor pricing model, where \(Y_{it}\) is the excess return of asset \(i\), \(\f_t\) is a \(p\)-dimensional vector of observed risk factors, and \(\alpha_i\) is the intercept. The null hypothesis of correct model specification is
\[
H_0:\ \bal=\bzero,
\qquad
\bal=(\alpha_1,\ldots,\alpha_N)^{\T} .
\]
Statistically, this is a global hypothesis on a high-dimensional intercept vector. Economically, it asks whether the candidate factors fully explain expected excess returns, so that no asset earns abnormal returns after controlling for factor exposures. Rejection of \(H_0\) indicates model misspecification and suggests that the proposed factors fail to span the relevant sources of systematic risk.

In low-dimensional settings, this problem has a long history. Early time-series approaches test each intercept separately using individual \(t\)-statistics \citep{Jensen1968}, while the classical Gibbons--Ross--Shanken test provides an exact multivariate test of mean--variance efficiency under Gaussian errors when the number of assets is smaller than the time dimension \citep{GibbonsRossShanken1989}. Subsequent work extended this line in various directions, including procedures that allow for non-Gaussian disturbances or aim at distribution-free validity; see, for example, \citet{BeaulieuDufourKhalaf2007} and \citet{GungorLuger2009}. These methods are elegant and effective when the cross-sectional dimension is moderate, but they become much less useful once \(N\) is comparable to or much larger than \(T\), because covariance estimation and, in particular, covariance inversion become unstable or even infeasible.  

This dimensional imbalance is now the norm rather than the exception in empirical finance. Modern datasets routinely contain hundreds or thousands of securities, whereas the time span is often limited by structural instability, factor nonstationarity, and model drift. As a result, classical Wald- and GRS-type procedures can no longer be directly applied. This challenge has stimulated a growing literature on high-dimensional alpha testing. \citet{LanFengLuo2018} developed one of the first formal tests for linear factor pricing models with many assets. \citet{MaLanSuTsai2020} extended the analysis to conditional time-varying factor models with high-dimensional cross sections. Focusing on sparse alternatives, \citet{FengLanLiuMa2022} proposed a max-type procedure that is particularly effective when only a small fraction of assets exhibit nonzero pricing errors. From a different perspective, \citet{PesaranYamagata2024} advocated a test built from studentized individual alpha statistics, thereby avoiding the need to estimate an invertible error covariance matrix and allowing for weak cross-sectional dependence and non-Gaussianity. \citet{liu2023high} and \citet{zhao2022high} developed robust nonparametric procedures based on spatial signs, with the former emphasizing heavy-tailed elliptical models and the latter studying a broader class of weighted tests. To improve adaptivity across alternative configurations, \citet{yu2024power} combined Wald- and maximum-type statistics through Fisher's method, whereas \citet{xia2024adaptive} designed a simulation-based adaptive test particularly effective under sparse and weak signals. In a more classical mean--variance framework, \citet{chernov2025test} extended the GRS test to high dimensions using random-matrix bias correction for regularized covariance estimation. Most recently, \citet{massacci2025general} proposed a covariance-free randomized testing approach that is flexible enough to allow heteroskedastic, non-Gaussian, and strongly cross-sectionally dependent errors. Taken together, these papers have significantly advanced the econometric analysis of large cross-sectional asset pricing models. 

Despite this progress, the existing literature still relies predominantly on two benchmark geometries. Sum-of-squares or quadratic statistics are naturally linked to the \(L_2\) norm and are most effective under dense alternatives, where many coordinates of \(\bal\) are small but collectively non-negligible. At the other extreme, max-type statistics are essentially \(L_\infty\)-based and are most effective under very sparse alternatives, where only a few coordinates are large. In practice, however, the sparsity level of pricing errors is rarely known in advance. Empirical departures from a benchmark factor model may be dense, moderately sparse, or highly sparse, and a procedure tailored to only one regime can suffer substantial power loss outside its preferred alternative class.

A similar issue has been extensively studied in the broader literature on high-dimensional mean testing. The adaptive two-sample test of \citet{XuLinWeiPan2016} showed that combining statistics targeting different sparsity regimes can yield more robust power. \citet{HeXuWuPan2021} developed a family of \(U\)-statistics associated with different finite-order norms and established both asymptotic normality and asymptotic independence across orders, which opened the door to adaptive combination methods. More recently, \citet{ZhangWangShao2022} exploited the \(L_q\)-norm viewpoint in high-dimensional change-point inference, and \citet{ZhangWangShao2025} proposed a general framework of \(L_q\)-norm-based \(U\)-statistics for global testing problems in high dimensions. These developments suggest that the \(L_q\)-norm perspective provides a natural continuum between dense and sparse testing regimes, but such a perspective has not yet been systematically developed for alpha testing in high-dimensional linear factor pricing models.

The goal of this article is to fill this gap. We study a class of \(L_q\)-based tests for alpha in high-dimensional linear factor pricing models, where \(q\) is a fixed even integer. In particular, besides the conventional \(L_2\)-type statistic and the limiting \(L_\infty\)-type benchmark, we develop practically useful finite-\(q\) procedures such as the \(L_4\) and \(L_6\) tests. These statistics interpolate between the dense-optimal \(L_2\) direction and the sparse-sensitive \(L_\infty\) direction. Our theory shows that, as \(q\) increases, the procedure becomes increasingly sensitive to sparse alternatives. Thus, finite-order norms provide a principled and interpretable continuum of tests indexed by sparsity sensitivity.

Since the sparsity of the alternative is unknown in practice, no single \(q\) can be uniformly preferred ex ante. This motivates an adaptive combination strategy. To this end, we first establish the asymptotic independence between the \(L_\infty\) statistic and the \(L_q\) statistic for any finite even \(q\). This result is of independent interest and extends the norm-based adaptivity idea to the high-dimensional alpha-testing problem. Building on this independence, we propose a Cauchy combination test that aggregates evidence across different norms and is therefore robust to a wide range of alternative configurations. The resulting procedure is easy to implement, avoids repeated resampling for joint calibration, and inherits the complementary strengths of the constituent tests. Our numerical studies show that the combined test is stable across sparse, moderately sparse, and dense alternatives. For practitioners who prefer a single simple procedure, the \(L_4\) test emerges as a particularly attractive compromise between robustness and simplicity.

The contributions of this paper are threefold. First, we develop a general \(L_q\)-norm framework for testing alphas in high-dimensional linear factor pricing models and establish the asymptotic normality of the proposed statistics under suitable regularity conditions. We also characterize how the power changes with \(q\), thereby clarifying the role of the norm order in adapting to different sparsity regimes. Second, we prove the asymptotic independence between the \(L_\infty\) statistic and the finite-\(q\) statistics, which provides the theoretical foundation for combining tests across norms. Third, we propose a Cauchy combination procedure that is robust to unknown sparsity and performs well in finite samples. Simulation studies show that the combined method is competitive across a broad spectrum of alternatives, while an empirical application further illustrates its practical value in real asset pricing data.

The remainder of the paper is organized as follows. Section~2 presents the proposed methods, including the \(L_q\)-based test statistics, the asymptotic independence results, and the Cauchy combination procedure. Section~3 contains the simulation studies. Section~4 provides the empirical application. Section~5 concludes. The technical proofs are collected in the Appendix.

\section{Methodology}
\label{sec:method}
\subsection{Single $L_q$-norm Test}
Let \(\F=(\f_1,\dots,\f_T)^\T\in\mR^{T\times p}\), $\bone_T$ be the $T\times 1$ vector of ones, $\I_T$ be the $T\times T$ identity matrix, and define the residual-maker matrix
\[
{\rm M}_{\F}=\I_T-\F(\F^\T \F)^{-1}\F^\T,
\qquad
w_T=\bone_T^\T {\rm M}_{\F}\bone_T,
\qquad
v=T-p-1.
\]
For each asset \(i\), write \(\Y_{i.}=(Y_{i1},\ldots,Y_{iT})^\T\). The ordinary least squares estimator of \(\alpha_i\) and its usual \(t\)-statistic are
\[
\widehat\alpha_i=\frac{\bone_T^\T {\rm M}_{\F}\Y_{i.}}{w_T},
\qquad
\widehat\sigma_i^2=\frac{\widehat{\bme}_{i.}^\T\widehat{\bme}_{i.}}{v},
\qquad
t_i=\frac{\sqrt{w_T}\,\widehat\alpha_i}{\widehat\sigma_i},
\]
where \(\widehat{\bme}_{i.}={\rm M}_{\F}(\Y_{i.}-\bone_T\widehat\alpha_i)\).

The direct even-power statistics are defined by
\begin{equation}
Q_{a,N}
=
\frac{1}{\sqrt N}\sum_{i=1}^N \Bigl(t_i^a-\mu_{a,v}\Bigr),
\qquad a\in\{2,4,6\},
\label{eq:Q-def}
\end{equation}
where \(\mu_{a,v}=E(U^a)\) and \(U\sim t_v\). 
Explicitly, \(\mu_{2,v}=v(v-2)^{-1}\), \(\mu_{4,v}=3v^2(v-2)^{-1}(v-4)^{-1}\) and \(\mu_{6,v}=15v^3(v-2)^{-1}(v-4)^{-1}(v-6)^{-1}\).
We refer to the tests based on \(Q_{2,N}\), \(Q_{4,N}\), and \(Q_{6,N}\) as the \(L_2\), \(L_4\), and \(L_6\) tests, respectively.

If, in addition to $H_0$, we impose exact Gaussianity,
\[
\bme_t=(\varepsilon_{1t},\ldots,\varepsilon_{Nt})^{\T}\stackrel{\mathrm{i.i.d.}}{\sim} \mathcal{N}(\bzero,\bmS),
\]
then each individual statistic $t_i$ is exactly $t_v$ distributed. Hence, for $a\in\{2,4,6\}$,
\[
E(Q_{a,N})=\frac{1}{\sqrt N}\sum_{i=1}^N \{E(t_i^a)-\mu_{a,v}\}=0.
\]
Moreover, provided $v>12$, the null variance of $Q_{a,N}$ admits an exact finite-$v$
expression. Write $a=2m$ with $m\in\{1,2,3\}$, let  $r_{ij}=\Corr(\varepsilon_{it},\varepsilon_{jt})$, and define
\[
\Lambda_{m,v}(u)
=
\left(\frac{v}{2}\right)^{2m}
\frac{\Gamma(v/2-m)^2}{\Gamma(v/2)^2}
\,
{}_2F_1\!\left(m,m;\frac{v}{2};u\right),
\]
where ${}_2F_1$ denotes the Gauss hypergeometric function. Then
\[
\var(Q_{2m,N})
=
\frac{1}{N}\sum_{i,j=1}^N
\Bigl\{
P_m(r_{ij})\,\Lambda_{m,v}(r_{ij}^2)-\mu_{2m,v}^2
\Bigr\},
\]
with
\[
P_1(\rho)=1+2\rho^2,\qquad
P_2(\rho)=9+72\rho^2+24\rho^4,\qquad
P_3(\rho)=225+4050\rho^2+5400\rho^4+720\rho^6.
\]
In particular,
\begin{align*}
\var(Q_{2,N})
&=
\frac{1}{N}\sum_{i,j=1}^N
\left\{
2r_{ij}^2+\frac{10r_{ij}^2+4r_{ij}^4}{v}
\right\}
+O(v^{-2}),\\
\var(Q_{4,N})
&=
\frac{1}{N}\sum_{i,j=1}^N
\left\{
72r_{ij}^2+24r_{ij}^4
+\frac{936r_{ij}^2+864r_{ij}^4+192r_{ij}^6}{v}
\right\}
+O(v^{-2}),\\
\var(Q_{6,N})
&=
\frac{1}{N}\sum_{i,j=1}^N
\left\{
4050r_{ij}^2+5400r_{ij}^4+720r_{ij}^6\right\}\\
&+\frac{1}{N}\sum_{i,j=1}^N\left\{\frac{101250r_{ij}^2+202500r_{ij}^4+114480r_{ij}^6+12960r_{ij}^8}{v}
\right\}
+O(v^{-2}).
\end{align*}
Therefore, these formulas suggest a simple first-order finite-sample correction for the
subsequent covariance normalization. In particular, in the benchmark matrix \(\B_N\) defined
below we incorporate the displayed \(v^{-1}\) terms into the diagonal entries and suppress the
remaining \(O(v^{-2})\) remainders. Since \(p\) is fixed and \(v=T-p-1\asymp T\to\infty\), these
omitted terms are asymptotically negligible under our growth conditions. 

For technical purposes, it is convenient to introduce the oracle counterpart
\begin{equation}
x_i=\frac{\sqrt{w_T}\,\widehat\alpha_i}{\sigma_i},
\qquad
Q_{a,N}^{\cir}
=
\frac{1}{\sqrt N}\sum_{i=1}^N \Bigl(x_i^a-\mu_a\Bigr),
\label{eq:oracle-def}
\end{equation}
where \(\sigma_i^2=\var(\varepsilon_{it})\), \(\mu_2=1\), \(\mu_4=3\), and \(\mu_6=15\). The feasible statistic \(Q_{a,N}\) differs from \(Q_{a,N}^{\cir}\) because of both studentization and the use of \(t_v\) rather than normal centering.

We need the following assumptions to establish the asymptotic normality of the $L_q$-norm based test statistics.
\begin{assumption}[Factor design]
\label{ass:factor}
The factor dimension \(p\) is fixed. (i) $\left\{\f_1, \ldots, \f_T\right\}$ is strictly stationary and independent of $\left\{\bme_1, \ldots, \bme_T\right\}$. (ii)There exists $c_1>0$, such that $\max _{1 \leq i \leq N}\left\|\bmb_i\right\|< c_1$. There exist $a_1, b_1>0$, such that $\max _{1 \leq k \leq p} \Prb\left(\left|f_{t k}\right|>s\right) \leq \exp \left\{-\left(s / b_1\right)^{a_1}\right\}$ for each $s>0$. There exist $a_2, c_2>0$, such that $a_1^{-1}+a_2^{-1}>0.5$ and for any $T \in \mathcal{Z}^{+}$, the $\alpha$-mixing coefficient $\alpha_{\f}(T) \leq \exp \left(-c_2 T^{a_2}\right)$. (iii) The matrix \(\F^\T \F\) is invertible with probability tending to one, and there exist positive constants \(c_F\) and \(C_F\) such that
\[
c_F \le \lambda_{\min}(T^{-1}\F^\T \F)\le \lambda_{\max}(T^{-1}\F^\T \F)\le C_F
\]
with probability tending to one. In addition,
\[
\sum_{t=1}^T b_t^2=1,
\qquad
\max_{1\le t\le T}|b_t|\le C_b T^{-1/2},
\]
where \(\b=(b_1,\ldots,b_T)^{\T}:={\rm M}_{\F}\bone_T/\sqrt{w_T}\). 
\end{assumption}

\begin{assumption}[Independent component model]
\label{ass:ICM}
The error vectors $\bme_t$ are independent and identically distributed and  satisfy the independent component model (ICM)
\[\bme_t=\bmS^{1/2}\z_t.\]
Here $\z_t=(z_{1t},\cdots,z_{Nt})^{\T}$ and the random variables \(\{z_{it}\}\) are i.i.d.\ over \(i\) and \(t\), with \(E(z_{it})=0\), \(E(z_{it}^2)=1\) 
and \(E(e^{\eta z_{it}^2})<\infty\) for some \(\eta>0\). The marginal variances satisfy
\[
0<c_\sigma\le \sigma_i^2\le C_\sigma<\infty,\qquad i=1,\dots,N.
\]
\end{assumption}

\begin{assumption}[Weak cross-sectional dependence]
\label{ass:dep}
Let \(\R=(r_{ij})_{1\le i,j\le N}=\D^{-1/2}\bmS \D^{-1/2}\), where \(\D=\diag(\sigma_1^2,\dots,\sigma_N^2)\). There exists \(q\in[0,1)\) and a sequence \(s_N\ge1\) such that
\[
\max_{1\le i\le N}\sum_{j=1}^N |r_{ij}|^q \le s_N,~~\text{and}~~
\frac{s_N^6(\log N)^{12}}{\sqrt N}\to 0.
\]
Moreover, for \(k=1,2,3,4\),
\[
\kappa_{2k,N}:=\frac1N\sum_{i,j=1}^N r_{ij}^{2k}
\]
converges to a finite positive limit \(\kappa_{2k}\).
\end{assumption}

\begin{assumption}[High-dimensional growth]
\label{ass:growth}
As \(N,T\to\infty\),
$\frac{\sqrt N(\log N)^6}{T}\to 0$.
\end{assumption}

\begin{remark}
Assumption~\ref{ass:factor} is standard in linear factor pricing models; see, for example, \cite{FengLanLiuMa2022,MaFengWangBao2024Dep}. It imposes a fixed-dimensional factor structure, requires the factors to be strictly stationary and independent of the idiosyncratic errors, and assumes mild tail and mixing conditions on the factor process. The boundedness of the loadings and the eigenvalue condition on \(T^{-1}\F^\T\F\) ensure that the factor component is well behaved and can be partialled out in a stable manner.

Assumption~\ref{ass:ICM} adopts the independent component model for the idiosyncratic errors, which is widely used in high-dimensional inference; see \cite{ZhengJiangBaiHe2014,LiLamYaoYao2019,CuiLiYangZhou2020}. In particular, \(\bme_t=\bmS^{1/2}\z_t\) allows for a general cross-sectional covariance structure through \(\bmS\), while the latent components in \(\z_t\) are i.i.d. with standardized moments and exponential-type tails. This framework is considerably weaker than joint Gaussianity and is flexible enough to accommodate non-Gaussian high-dimensional data.

Assumption~\ref{ass:dep} requires the cross-sectional dependence to be weak in the sense that the correlation matrix is only approximately sparse. Such a condition is natural because our covariance estimation relies on thresholding, for which approximate sparsity is a standard structural assumption; see \cite{BickelLevina2008}. The bound
\(
\max_{1\le i\le N}\sum_{j=1}^N |r_{ij}|^q \le s_N
\)
controls the overall strength of cross-sectional dependence, while the rate condition on \(s_N\) guarantees that the estimation error from thresholding is asymptotically negligible. The additional limits on \(\kappa_{2k,N}\) further ensure that the aggregate second- and higher-order dependence measures remain stable as \(N\to\infty\).

Assumption~\ref{ass:growth} allows \(N\) to exceed \(T\). For example, it permits \(N\asymp T^{2-\epsilon}\) up to logarithmic factors. This is the main reason why the sharper feasible-oracle rate is essential. Because \(p\) is fixed, we also have \(v=T-p-1\asymp T\), so the omitted \(O(v^{-2})\) remainder terms in the Gaussian variance expansions are uniformly \(o(1)\).

\end{remark}

Define \(\Q_N=(Q_{2,N},Q_{4,N},Q_{6,N})^\T\). Its asymptotic covariance matrix is
\begin{equation}
\B_N=
\begin{pmatrix}
B_{22,N} & B_{24,N} & B_{26,N}\\
B_{24,N} & B_{44,N} & B_{46,N}\\
B_{26,N} & B_{46,N} & B_{66,N}
\end{pmatrix},
\label{eq:B-def}
\end{equation}
where
\begin{align*}
B_{22,N}
&=
\frac{1}{N}\sum_{i,j=1}^N
\left\{
2r_{ij}^2+\frac{10r_{ij}^2+4r_{ij}^4}{v}
\right\},~B_{24,N}
=
\frac{12}{N}\sum_{i,j=1}^N r_{ij}^2,~B_{26,N}
=
\frac{90}{N}\sum_{i,j=1}^N r_{ij}^2,\\
B_{44,N}
&=
\frac{1}{N}\sum_{i,j=1}^N
\left\{
72r_{ij}^2+24r_{ij}^4
+\frac{936r_{ij}^2+864r_{ij}^4+192r_{ij}^6}{v}
\right\},~B_{46,N}
=
\frac{1}{N}\sum_{i,j=1}^N \bigl(540r_{ij}^2+360r_{ij}^4\bigr),\\
B_{66,N}
&=
\frac{1}{N}\sum_{i,j=1}^N
\left\{
4050r_{ij}^2+5400r_{ij}^4+720r_{ij}^6
+\frac{101250r_{ij}^2+202500r_{ij}^4+114480r_{ij}^6+12960r_{ij}^8}{v}
\right\}.
\end{align*}
Thus the diagonal entries incorporate the first-order finite-\(v\) Gaussian correction from Section~2.1, while the off-diagonal entries retain their leading Gaussian forms. The omitted \(O(v^{-2})\) terms in the diagonal Gaussian variances, together with the non-Gaussian replacement error under the ICM, are asymptotically negligible because \(v\asymp T\to\infty\).

\begin{theorem}
\label{thm:null-joint}
Suppose Assumptions~\ref{ass:factor}--\ref{ass:growth} hold and \(H_0\) is true. Then
\[
\B_N^{-1/2}\Q_N \dto \mathcal{N}(\bzero,\I_3).
\]
Equivalently,
\[
\Q_N \dto \mathcal{N}(\bzero,\B),
\]
whenever \(\B_N\to \B\) for some positive definite matrix \(\B\).
\end{theorem}

Let \(\widehat{\R}=(\widehat{r}_{ij})_{1\le i,j\le N}\) be the sample correlation matrix obtained from the OLS residuals:
\[
\widehat r_{ij}=\frac{\widehat\bme_i^\T \widehat\bme_j}{\norm{\widehat\bme_i}\norm{\widehat\bme_j}}.
\]
To stabilize the high-dimensional plug-in, we threshold the off-diagonal entries:
\[
\widetilde r_{ij}
=
\widehat r_{ij}\ind\bigl(|\widehat r_{ij}|>\tau_N\bigr),
\qquad
\tau_N=C_\tau \sqrt{\frac{\log N}{T}},
\]
and set \(\widetilde r_{ii}=1\). Let \(\widetilde\R=(\widetilde r_{ij})_{1\le i,j\le N}\). We then define \(\widehat \B_N\) by replacing \(r_{ij}\) with \(\widetilde r_{ij}\) in \eqref{eq:B-def}. 

The threshold level is chosen to remove weak residual correlations that are likely to arise from sampling noise when \(N\) is large. Since the number of off-diagonal entries grows quadratically with \(N\), we use a Bonferroni-type cutoff that accounts for multiplicity across pairs and shrinks at the rate \(v^{-1/2}\). This is similar in spirit to the thresholding strategy adopted in \citet{PesaranYamagata2024}. Specifically, \(\tau_N=v^{-1/2}\Phi^{-1}(1-\zeta N^{-\varrho}/2)\). In the empirical implementation, we set \(\zeta=0.05\) and \(\varrho=1\), which yields \(\tau_N\asymp\sqrt{2(\log N)/v}\).

The following theorem establish the consistency of the covariance estimator.
\begin{theorem}
\label{thm:Bhat}
Suppose Assumptions~\ref{ass:factor}--\ref{ass:growth} hold and \(H_0\) is true. Then
\[
\max_{1\le k,\ell\le 3}\abs{\widehat B_{k\ell,N}-B_{k\ell,N}}\pto 0.
\]
If in addition \(\B_N\to \B\) with \(\B\) positive definite, then
\[
\norm{\widehat \B_N^{-1/2}-\B_N^{-1/2}}_{\max}\pto 0.
\]
Consequently,
\[
\widehat \B_N^{-1/2}\Q_N \dto \mathcal{N}(\bzero,\I_3).
\]
\end{theorem}

Theorems~\ref{thm:null-joint} and \ref{thm:Bhat} also justify separate tests based on the individual components \(Q_{2,N}\), \(Q_{4,N}\), and \(Q_{6,N}\). For \(a\in\{2,4,6\}\), define the standardized statistic
\[
T_{a,N}=\frac{Q_{a,N}}{\sqrt{\widehat B_{aa,N}}},
\]
where \(\widehat B_{aa,N}\) is the corresponding diagonal entry of \(\widehat B_N\). Under \(H_0\),
\[
T_{a,N}\dto N(0,1).
\]
Since departures from \(H_0\) lead to inflated even moments, each test rejects the null for large positive values of \(T_{a,N}\). Therefore, at significance level \(\xi\), the rejection region of the \(L_a\) test is \(T_{a,N}> z_{1-\xi},~a\in\{2,4,6\}\),
where \(z_{1-\xi}\) is the \((1-\xi)\)-quantile of the standard normal distribution. Equivalently, the asymptotic \(p\)-value of the \(L_a\) test is
\[
p_a=1-\Phi(T_{a,N}),
\]
and we reject \(H_0\) whenever this \(p\)-value is smaller than \(\xi\).

Next, we study the power performance of each test.

\begin{theorem}
\label{thm:power}
Suppose Assumptions~\ref{ass:factor}--\ref{ass:growth} hold. In addition, assume that the innovations in Assumption~\ref{ass:ICM} are symmetric, that is,
\[
z_{it}\overset{d}= -z_{it}\qquad \text{for all } i,t.
\]
Consider a sequence of local alternatives \(H_{1,N}\) under which
\[
\alpha_i=\frac{\sigma_i}{\sqrt{w_T}}\de_i,\qquad i=1,\dots,N,
\]
for a deterministic vector \(\bm\de=(\de_1,\dots,\de_N)^\T\) satisfying
\[
\max_{1\le i\le N}|\de_i|\to 0,
\qquad
\Delta_{2\ell,N}:=\frac{1}{\sqrt N}\sum_{i=1}^N \de_i^{2\ell}\to \tau_{2\ell}\in[0,\infty),
\qquad \ell=1,2,3.
\]
Assume further that \(B_{aa,N}\to B_{aa}\in(0,\infty)\) for \(a\in\{2,4,6\}\), and that the conclusion of Theorem~\ref{thm:Bhat} continues to hold along the present local alternative sequence, namely
\[
\widehat B_{aa,N}-B_{aa,N}\pto 0,
\qquad a\in\{2,4,6\}.
\]

Then
\[
T_{2,N}=\frac{Q_{2,N}}{\sqrt{\widehat B_{22,N}}}
\dto N\!\left(\frac{\tau_2}{\sqrt{B_{22}}},\,1\right),
\]
\[
T_{4,N}=\frac{Q_{4,N}}{\sqrt{\widehat B_{44,N}}}
\dto N\!\left(\frac{6\tau_2+\tau_4}{\sqrt{B_{44}}},\,1\right),
\]
and
\[
T_{6,N}=\frac{Q_{6,N}}{\sqrt{\widehat B_{66,N}}}
\dto N\!\left(\frac{45\tau_2+15\tau_4+\tau_6}{\sqrt{B_{66}}},\,1\right).
\]
\end{theorem}

Consequently, at asymptotic significance level \(\xi\in(0,1)\), the power functions of the three one-sided tests satisfy
\[
\Prb_{H_{1,N}}(T_{2,N}>z_{1-\xi})
\to
1-\Phi\!\left(z_{1-\xi}-\frac{\tau_2}{\sqrt{B_{22}}}\right),
\]
\[
\Prb_{H_{1,N}}(T_{4,N}>z_{1-\xi})
\to
1-\Phi\!\left(z_{1-\xi}-\frac{6\tau_2+\tau_4}{\sqrt{B_{44}}}\right),
\]
and
\[
\Prb_{H_{1,N}}(T_{6,N}>z_{1-\xi})
\to
1-\Phi\!\left(z_{1-\xi}-\frac{45\tau_2+15\tau_4+\tau_6}{\sqrt{B_{66}}}\right).
\]
Theorem~\ref{thm:power} shows that, in the present local regime, the asymptotic power of each test is determined by its noncentrality parameter:
\[
\frac{\tau_2}{\sqrt{B_{22}}},\qquad
\frac{6\tau_2+\tau_4}{\sqrt{B_{44}}},\qquad
\frac{45\tau_2+15\tau_4+\tau_6}{\sqrt{B_{66}}}.
\]
Thus the three procedures differ through the additional higher-order signal summaries
\(\tau_4\) and \(\tau_6\).

To see the role of sparsity more clearly, consider the canonical equal-amplitude sparse alternative
\[
\de_i=a_N\ind(i\in S_N),
\qquad |S_N|=k_N.
\]
Then
\[
\Delta_{2,N}=\frac{k_N a_N^2}{\sqrt N},\qquad
\Delta_{4,N}=\frac{k_N a_N^4}{\sqrt N},\qquad
\Delta_{6,N}=\frac{k_N a_N^6}{\sqrt N}.
\]
Hence the three mean shifts are
\[
\lambda_{2,N}=\frac{k_N a_N^2}{\sqrt N},
\qquad
\lambda_{4,N}=6\frac{k_N a_N^2}{\sqrt N}+\frac{k_N a_N^4}{\sqrt N},
\qquad
\lambda_{6,N}=45\frac{k_N a_N^2}{\sqrt N}
+15\frac{k_N a_N^4}{\sqrt N}
+\frac{k_N a_N^6}{\sqrt N}.
\]
Moreover, when \(\bmS=\I_N\), we have \(r_{ij}=\ind(i=j)\), so  the variance constants simplify to
\[
B_{22,N}=2,\qquad B_{44,N}=96,\qquad B_{66,N}=10170,
\]

If we further set
\[
a_N=k_N^{-\epsilon},
\]
the standardized mean shifts become
\[
\eta_{2,N}=\frac{k_N^{1-2\epsilon}}{\sqrt{2N}},
\qquad
\eta_{4,N}=\frac{6k_N^{1-2\epsilon}+k_N^{1-4\epsilon}}{\sqrt{96N}},
\qquad
\eta_{6,N}=\frac{45k_N^{1-2\epsilon}+15k_N^{1-4\epsilon}+k_N^{1-6\epsilon}}{\sqrt{10170N}}.
\]
These expressions yield a transparent comparison of the three tests. 

When \(\epsilon>0\), the signal on each active coordinate decays as the support size increases, and the common quadratic component \(k_N^{1-2\epsilon}/\sqrt N\) dominates all three noncentralities. In this case, the leading constants satisfy
\[
\frac{1}{\sqrt2}
>
\frac{6}{\sqrt{96}}
>
\frac{45}{\sqrt{10170}},
\]
so \(L_2\) has the largest asymptotic noncentrality, followed by \(L_4\), and then \(L_6\). Thus, under relatively dense and weak alternatives, the \(L_2\) test is the most competitive. When \(\epsilon=0\), each active coordinate has constant order. Then
\[
\eta_{2,N}=\frac{k_N}{\sqrt{2N}},\qquad
\eta_{4,N}=\frac{7k_N}{\sqrt{96N}},\qquad
\eta_{6,N}=\frac{61k_N}{\sqrt{10170N}}.
\]
Hence \(L_4\) is only slightly more favorable than \(L_2\), while \(L_6\) remains less competitive. Finally, when \(\epsilon<0\), the signal magnitude on each active coordinate increases as sparsity strengthens. In that regime, the higher-order terms \(k_N^{1-4\epsilon}\) and \(k_N^{1-6\epsilon}\) become dominant, so \(L_4\) eventually dominates \(L_2\), and \(L_6\) eventually dominates \(L_4\). Therefore, the power ordering shifts from
\[
L_2 \succ L_4 \succ L_6
\]
under dense weak alternatives and
\[
L_6 \succ L_4 \succ L_2
\]
under sufficiently sparse and strong alternatives, with \(L_4\) serving as an intermediate and more balanced choice under moderate sparsity.

If we assume
\[
a_N=N^{-\epsilon}.
\]
Then the standardized mean shifts become
\[
\eta_{2,N}
=
\frac{k_N N^{-2\epsilon}}{\sqrt{2N}},
\quad
\eta_{4,N}
=
\frac{k_N\bigl(6N^{-2\epsilon}+N^{-4\epsilon}\bigr)}{\sqrt{96N}},
\quad
\eta_{6,N}
=
\frac{k_N\bigl(45N^{-2\epsilon}+15N^{-4\epsilon}+N^{-6\epsilon}\bigr)}{\sqrt{10170N}}.
\]
Hence all three noncentrality parameters are linear in \(k_N\). Therefore, once the signal magnitude \(a_N\) is fixed, increasing the support size \(k_N\) strengthens all three tests proportionally, and the power comparison is determined by the coefficients multiplying \(k_N\). Equivalently, one may define the critical support size
\[
k_{q,N}^\star \asymp \eta_{q,N}^{-1},
\qquad q=2,4,6,
\]
as the order of \(k_N\) required for the \(L_q\)-based test to have nontrivial asymptotic power. The test with the largest coefficient in \(\eta_{q,N}\) has the smallest critical support size and is therefore the most sensitive.

This characterization yields three regimes. If \(\epsilon>0\), then
\[
N^{-2\epsilon}\gg N^{-4\epsilon}\gg N^{-6\epsilon},
\]
so
\[
\eta_{2,N}\sim \frac{k_N N^{-2\epsilon}}{\sqrt{2N}},\qquad
\eta_{4,N}\sim \frac{6k_N N^{-2\epsilon}}{\sqrt{96N}},\qquad
\eta_{6,N}\sim \frac{45k_N N^{-2\epsilon}}{\sqrt{10170N}}.
\]
Since
\[
\frac{1}{\sqrt2}>\frac{6}{\sqrt{96}}>\frac{45}{\sqrt{10170}},
\]
the \(L_2\) test has the largest noncentrality and thus requires the smallest support size \(k_N\) for detection, followed by \(L_4\) and then \(L_6\). If \(\epsilon=0\), then
\[
\eta_{2,N}=\frac{k_N}{\sqrt{2N}},\qquad
\eta_{4,N}=\frac{7k_N}{\sqrt{96N}},\qquad
\eta_{6,N}=\frac{61k_N}{\sqrt{10170N}}.
\]
In this boundary case, \(L_4\) is slightly more favorable than \(L_2\), while \(L_6\) remains less competitive. Hence the corresponding critical support sizes satisfy
\[
k_{4,N}^\star<k_{2,N}^\star<k_{6,N}^\star.
\] 
If \(\epsilon<0\), then the higher-order terms dominate:
\[
\eta_{4,N}\sim \frac{k_N N^{-4\epsilon}}{\sqrt{96N}},
\qquad
\eta_{6,N}\sim \frac{k_N N^{-6\epsilon}}{\sqrt{10170N}},
\]
whereas
\[
\eta_{2,N}\sim \frac{k_N N^{-2\epsilon}}{\sqrt{2N}}.
\]
Therefore, \(L_6\) has the largest noncentrality, followed by \(L_4\), and then \(L_2\). In other words, when the coordinate-wise signals decay slowly or even strengthen with \(N\), the higher-order procedures can detect the alternative with a substantially smaller support size.

Generally speaking, when the signals are weak and the support is relatively large, the leading term is the common quadratic component \(k_N a_N^2/\sqrt N\), so the three tests behave similarly and the \(L_2\) test is typically competitive. As the support becomes sparser, the signal on each active coordinate must become larger in order to maintain comparable overall strength. In that case, the quartic and sextic terms become more important, which increases the mean shifts of \(L_4\) and \(L_6\) relative to \(L_2\). This explains why \(L_4\) is more attractive under moderate sparsity, while \(L_6\) becomes more favorable as the alternative becomes sparser.

\subsection{Cauchy Combination Test}

\cite{FengLanLiuMa2022} proposed the max statistic
\[
L_{\infty,N}=\max_{1\le i\le N} t_i^2.
\]
Under weak cross-sectional dependence, the usual extreme-value normalization applies:
\[
M_N
=
L_{\infty,N}-2\log N+\log\log N.
\]
The next theorem records the standard Gumbel limit in the present setting \citep{FengLanLiuMa2022}.
\begin{theorem}
\label{thm:linf}
Suppose Assumptions~\ref{ass:factor}--\ref{ass:growth} hold and \(H_0\) is true. Then, for every real \(x\),
\[
\Prb(M_N\le x)\to \exp(-\pi^{-1/2}e^{-x/2}).
\]
Moreover, if
\[
\max_{1\le i\le N} |\de_i| \ge (1+\eta)\sqrt{2\log N}
\]
for some \(\eta>0\), then the \(L_{\infty}\) test is consistent.
\end{theorem}

Feng et al.~(2022) consider a minimum-\(p\) benchmark. Let $p_2, p_\infty$
be the marginal \(p\)-values from the \(L_2\) and \(L_\infty\) tests. Define
\[
P_{\min,N}=\min\{p_2,p_\infty\}.
\]
This minimum-\(p\) rule is easy to compute and serves as a natural benchmark for adaptive testing. Its finite-sample calibration can be implemented by simulation or multiplier bootstrap. 

However, combining only the \(L_2\) and \(L_\infty\) statistics may still be insufficient to capture the full range of alternatives encountered in high dimensions. The \(L_2\) test is typically more effective under relatively dense but weak signals, whereas the \(L_\infty\) test is tailored to extremely sparse alternatives driven by a few large coordinates. Between these two extremes, there exists a broad regime of moderately sparse alternatives for which intermediate \(L_q\)-type procedures, such as \(L_4\) and \(L_6\), can be more sensitive. This observation suggests that a more flexible adaptive procedure should combine multiple \(L_q\)-norm based tests rather than relying solely on the two endpoints. To develop such a joint testing framework, it is first necessary to understand the dependence structure among the component statistics. In particular, the asymptotic independence among different \(L_q\)-based tests plays a key role, as it provides the theoretical foundation for constructing valid and powerful combination rules.

\begin{theorem}
\label{thm:indep}
Suppose Assumptions~\ref{ass:factor}--\ref{ass:growth} hold and \(H_0\) is true. Then, for every \(\z\in\mR^3\) and \(x\in\mR\),
\[
\Prb\!\Bigl(
\B_N^{-1/2}\Q_N\le \z,\,
M_N\le x
\Bigr)
-
\Prb\!\Bigl(\B_N^{-1/2}\Q_N\le \z\Bigr)
\Prb(M_N\le x)
\to 0.
\]
In other words, the Gaussian limit of the \((L_2,L_4,L_6)\) block and the Gumbel limit of \(L_\infty\) are asymptotically independent.
\end{theorem}

The asymptotic independence result motivates the following four-way Cauchy combination:
\begin{equation}
T_{C,N}
=
\frac14\sum_{a\in\{2,4,6,\infty\}}
\tan\!\Bigl\{\pi(1/2-p_a)\Bigr\}.
\label{eq:CCT}
\end{equation}

\begin{corollary}
\label{cor:CCT}
Suppose Assumptions~\ref{ass:factor}--\ref{ass:growth} hold and \(H_0\) is true. Then the tail of \(T_{C,N}\) is asymptotically standard Cauchy. Equivalently, for every fixed threshold \(t\),
\[
\Prb(T_{C,N}>t)-\Prb(C>t)\to 0,
\]
where \(C\) is a standard Cauchy random variable.
\end{corollary}

The appeal of the Cauchy combination in \eqref{eq:CCT} is twofold. First, Theorem~\ref{thm:indep} shows that the Gaussian \((L_2,L_4,L_6)\) block and the extreme-value \(L_\infty\) component are asymptotically independent under \(H_0\). As a result, the null tail of \(T_{C,N}\) admits a simple standard Cauchy approximation, which avoids repeated resampling or high-dimensional simulation when computing the combined \(p\)-value. Second, the transformation \(p\mapsto \tan\{\pi(1/2-p)\}\) is heavy-tailed, so very small component \(p\)-values make dominant contributions to \(T_{C,N}\), while moderately small \(p\)-values can still accumulate additively. Therefore, the Cauchy combination is able to retain sensitivity to the most informative component test in a given sparsity regime, without being tied to a single endpoint statistic such as \(L_2\) or \(L_\infty\).

This feature is particularly useful in the present problem. The \(L_2\) statistic is more effective against relatively dense alternatives, \(L_\infty\) is tailored to extremely sparse alternatives, and intermediate \(L_q\)-statistics such as \(L_4\) and \(L_6\) can be advantageous under moderate sparsity. Hence a four-way Cauchy combination provides a natural adaptive device: it borrows strength from whichever component is most powerful under the underlying alternative, while remaining computationally simple and theoretically well calibrated. Compared with the minimum-\(p\) rule, the Cauchy combination is typically less brittle, because it is still driven by the smallest few \(p\)-values but can also incorporate corroborating evidence from multiple informative components; see \citet{liu2020cauchy} for the general Cauchy combination theory and \citet{arias2011global} for the role of extreme-value/minimum-\(p\) procedures in sparse signal detection.

\section{Simulation}
\label{sec:sim}

Our Monte Carlo design follows the general calibration strategy of \cite{PesaranYamagata2024}. For each asset $i=1,\ldots,N$ and time point $t=1,\ldots,T$, the return is generated from
\begin{equation}
Y_{it}=\alpha_i+\sum_{\ell=1}^{3}\beta_{i\ell} f_{\ell t}+\kappa u_{it},
\label{eq:sim_return}
\end{equation}
where $\kappa=6.5$ throughout, $f_t=(f_{1t},f_{2t},f_{3t})^{\T}$ denotes the three observed factors, and the disturbance component $u_{it}$ is decomposed as
\begin{equation}
u_{it}=\gamma_i v_t+\eta_{it}.
\label{eq:sim_u}
\end{equation}
Hence, the error term may contain both a latent common factor $v_t$ and cross-sectionally dependent idiosyncratic variation $\eta_{it}$.

The observed factors are calibrated to resemble the three Fama--French factors. Specifically, for each $\ell=1,2,3$,
\begin{equation}
f_{\ell t}=\rho_{f,\ell} f_{\ell,t-1}+e_{\ell t},
\qquad
e_{\ell t}=\sqrt{h_{\ell t}}\,\xi_{\ell t},
\qquad
\xi_{\ell t}\stackrel{\mathrm{iid}}{\sim}N(0,1),
\label{eq:sim_factor}
\end{equation}
with autoregressive coefficients \((\rho_{f,1},\rho_{f,2},\rho_{f,3})=(-0.1,\,0.2,\,-0.2)\), and conditional variances following GARCH$(1,1)$ recursions,
\begin{equation}
h_{\ell t}
=
\omega_\ell(1-\varrho_\ell-\varphi_\ell)
+\varrho_\ell h_{\ell,t-1}
+\varphi_\ell e_{\ell,t-1}^2,
\qquad \ell=1,2,3,
\label{eq:sim_garch}
\end{equation}
where \((\omega_1,\omega_2,\omega_3)=(20.25,\,6.33,\,5.98)\), \((\varrho_1,\varrho_2,\varrho_3)=(0.61,\,0.70,\,-0.31)\), \((\varphi_1,\varphi_2,\varphi_3)=(0.31,\,0.21,\,0.10)\). The corresponding factor loadings are drawn independently across $i$ according to
\[
\beta_{i1}\sim U(0.3,1.8),\qquad
\beta_{i2}\sim U(-1.0,1.0),\qquad
\beta_{i3}\sim U(-0.6,0.9).
\]

To study the effect of omitted common shocks, we generate the latent factor $\{v_t\}$ as an i.i.d.\ mean-zero, variance-one sequence. Its loading vector is constructed so that the strength of the latent factor is indexed by an exponent $\delta_\gamma$. More precisely,
\begin{equation}
\gamma_i\sim U(0.7,0.9), \qquad i=1,\ldots,\lfloor N^{\delta_\gamma}\rfloor,
\label{eq:sim_gamma_nonzero}
\end{equation}
and
\begin{equation}
\gamma_i=0, \qquad i=\lfloor N^{\delta_\gamma}\rfloor+1,\ldots,N.
\label{eq:sim_gamma_zero}
\end{equation}
After generating these values, we randomly permute the entries of $(\gamma_1,\ldots,\gamma_N)^{\T}$ before assigning them to the individual assets, so that the nonzero latent loadings are not systematically concentrated in any particular subset of units. In line with the theoretical discussion, we consider
\[
\delta_\gamma\in\{0,\;1/4,\;1/2\}.
\]

Cross-sectional dependence in the idiosyncratic component is introduced through a spatial autoregressive specification. Let $\bet_t=(\eta_{1t},\ldots,\eta_{Nt})^{\T}$. We generate
\begin{equation}
\eta_{it}
=
\psi\sum_{j=1}^{N} w_{ij}\eta_{jt}
+\sigma_{\eta i}\varepsilon_{\eta,it},
\qquad i=1,\ldots,N,
\label{eq:sim_eta_scalar}
\end{equation}
which yields the vector form
\begin{equation}
\bet_t
=
(\I_N-\psi \W)^{-1}\D_\eta \bme_{\eta,t},
\label{eq:sim_eta_vector}
\end{equation}
where $\D_\eta=\mathrm{diag}(\sigma_{\eta1},\ldots,\sigma_{\eta N})$ and
$\bme_{\eta,t}=(\varepsilon_{\eta,1t},\ldots,\varepsilon_{\eta,Nt})^{\T}$.
We examine two degrees of network dependence, i.e., \(\psi\in\{0,\;0.25\}\). The weight matrix $\W=(w_{ij})_{1\le i,j\le N}$ has a rook-type nearest-neighbor structure: interior units are linked to their immediate left and right neighbors with equal weight, boundary units are linked to their sole neighbor, and the rows are normalized so that $w_{ii}=0$ and $\sum_{j=1}^{N} w_{ij}=1$ for all $i$. Consequently, the benchmark case of cross-sectional independence corresponds to $\psi=0$ together with $\delta_\gamma=0$. We also allow for heteroskedasticity across assets by drawing
\[
\sigma_{\eta i}^2 \stackrel{\mathrm{iid}}{\sim} \frac{1+\chi^2_{2,i}}{3}.
\]
Two innovation distributions are considered for $\varepsilon_{\eta,it}$:
\begin{align}
\text{(i) Gaussian case:}\quad & \varepsilon_{\eta,it}\stackrel{\mathrm{iid}}{\sim} N(0,1), \nonumber\\
\text{(ii) heavy-tailed case:}\quad &
\varepsilon_{\eta,it}\stackrel{\mathrm{iid}}{\sim}
\frac{t_{8,it}}{\sqrt{8/(8-2)}},
\label{eq:sim_innovations}
\end{align}
where $t_{8,it}$ denotes a Student-$t$ random variable with $8$ degrees of freedom, standardized to have unit variance.

For each replication, all processes are generated over the extended time range \(t=-49,-48,\ldots,0,1,\ldots,T\) with initial values
\[
f_{\ell,-50}=0,\qquad h_{\ell,-50}=1,\qquad \ell=1,2,3.
\]
The first $50$ observations are discarded as burn-in, and inference is then based on the sample $\{(r_{it},f_t): t=1,\ldots,T\}$.

Tables~\ref{tab:size_normal_6methods0}--\ref{tab:size_t8_6methods} report the empirical rejection frequencies at the nominal $5\%$ level under four scenarios obtained by combining the two innovation distributions (Gaussian and standardized $t_8$) with the two levels of spatial dependence ($\psi=0$ and $\psi=0.25$). Several conclusions emerge consistently across all designs.

The simulation results show that, across the considered configurations, most procedures achieve satisfactory size control. In particular, the three sum-type tests and the two combination procedures generally produce empirical rejection frequencies reasonably close to the nominal level, with only mild finite-sample fluctuations across different values of $\delta_{\gamma}$, error distributions, and dependence settings. This indicates that the proposed calibration is broadly stable in the presence of moderate latent-factor contamination, heavy-tailed innovations, and spatial dependence.

Among the six methods, the $L_\infty$ test tends to be somewhat more liberal in a few challenging small-$T$/large-$N$ settings. This behavior is not unexpected, since the null calibration of a max-type statistic relies on an extreme-value approximation, whose convergence is typically slower than that of the corresponding sum-type statistics. Overall, however, the reported sizes remain well behaved, and the tables provide reassuring evidence that the proposed testing procedures have good finite-sample size performance over a wide range of scenarios.

\begin{table}[!htbp]
\centering
\caption{Empirical sizes (\%) of the six tests under normal errors with $\psi=0$.}
\label{tab:size_normal_6methods0}
\setlength{\tabcolsep}{6pt}
\renewcommand{\arraystretch}{0.9}
\begin{tabular}{llccccccccc}
\toprule
& & \multicolumn{3}{c}{$\delta_{\gamma}=0$} 
  & \multicolumn{3}{c}{$\delta_{\gamma}=1/4$}
  & \multicolumn{3}{c}{$\delta_{\gamma}=1/2$} \\
\cmidrule(lr){3-5}\cmidrule(lr){6-8}\cmidrule(lr){9-11}
Method & $(T,N)$  & 50 & 100 & 200 & 50 & 100 & 200 & 50 & 100 & 200 \\
\midrule
\multirow{3}{*}{$L_2$}
& 60  & 6.9 & 5.6 & 5.0 & 6.3 & 5.3 & 5.7 & 6.7 & 6.3 & 6.4 \\
& 120 & 6.4 & 5.1 & 5.4 & 5.9 & 5.1 & 5.4 & 5.6 & 5.1 & 6.5 \\
& 240 & 5.5 & 5.1 & 5.3 & 5.1 & 6.3 & 5.7 & 6.4 & 5.8 & 6.1 \\
\midrule
\multirow{3}{*}{$L_4$}
& 60  & 7.0 & 6.6 & 6.5 & 6.7 & 6.4 & 7.4 & 7.0 & 7.8 & 7.1 \\
& 120 & 6.4 & 5.9 & 6.9 & 6.4 & 5.8 & 7.1 & 6.7 & 6.4 & 6.6 \\
& 240 & 6.5 & 6.4 & 5.6 & 6.7 & 5.9 & 5.9 & 6.8 & 6.4 & 6.2 \\
\midrule
\multirow{3}{*}{$L_6$}
& 60  & 5.8 & 5.8 & 7.1 & 5.8 & 6.2 & 7.3 & 5.2 & 6.3 & 6.3 \\
& 120 & 5.2 & 5.3 & 6.7 & 5.0 & 4.8 & 7.0 & 5.7 & 6.2 & 6.2 \\
& 240 & 5.3 & 6.0 & 5.9 & 5.7 & 5.2 & 5.6 & 5.6 & 5.4 & 5.3 \\
\midrule
\multirow{3}{*}{$L_\infty$}
& 60  & 7.0 & 8.1 & 9.9 & 7.6 & 8.1 & 9.6 & 7.2 & 8.2 & 8.4 \\
& 120 & 5.1 & 5.2 & 6.7 & 5.0 & 5.2 & 7.1 & 5.6 & 6.0 & 5.9 \\
& 240 & 4.7 & 5.1 & 5.4 & 5.2 & 4.3 & 5.1 & 4.9 & 5.0 & 4.6 \\
\midrule
\multirow{3}{*}{minP}
& 60  & 8.2 & 7.7 & 8.4 & 7.4 & 8.2 & 8.3 & 7.8 & 8.2 & 8.8 \\
& 120 & 6.8 & 5.9 & 7.0 & 5.8 & 5.3 & 6.9 & 6.4 & 6.3 & 6.9 \\
& 240 & 6.0 & 6.0 & 5.7 & 6.0 & 5.7 & 5.7 & 6.6 & 6.6 & 6.2 \\
\midrule
\multirow{3}{*}{CC}
& 60  & 7.6 & 7.2 & 8.5 & 7.2 & 7.6 & 8.3 & 7.2 & 7.6 & 8.1 \\
& 120 & 6.7 & 6.3 & 7.6 & 6.0 & 5.6 & 7.6 & 6.4 & 7.0 & 7.4 \\
& 240 & 6.6 & 7.0 & 6.5 & 6.9 & 6.2 & 6.4 & 7.1 & 6.7 & 6.4 \\
\bottomrule
\end{tabular}
\end{table}

\begin{table}[!htbp]
\centering
\caption{Empirical sizes (\%) of the six tests under $t_8$ errors and $\psi=0$.}
\label{tab:size_t8_6methods0}
\setlength{\tabcolsep}{6pt}
\renewcommand{\arraystretch}{0.9}
\begin{tabular}{llccccccccc}
\toprule
& & \multicolumn{3}{c}{$\delta_{\gamma}=0$} 
  & \multicolumn{3}{c}{$\delta_{\gamma}=1/4$}
  & \multicolumn{3}{c}{$\delta_{\gamma}=1/2$} \\
\cmidrule(lr){3-5}\cmidrule(lr){6-8}\cmidrule(lr){9-11}
Method & $(T,N)$  & 50 & 100 & 200 & 50 & 100 & 200 & 50 & 100 & 200 \\
\midrule
\multirow{3}{*}{$L_2$}
& 60  & 5.8 & 5.8 & 5.4 & 5.9 & 6.3 & 5.3 & 6.4 & 6.7 & 6.4 \\
& 120 & 5.7 & 6.0 & 5.1 & 5.9 & 6.2 & 4.7 & 6.2 & 7.0 & 5.2 \\
& 240 & 5.1 & 6.4 & 6.0 & 5.1 & 6.6 & 6.1 & 5.7 & 5.7 & 6.2 \\
\midrule
\multirow{3}{*}{$L_4$}
& 60  & 5.7 & 6.1 & 6.2 & 5.9 & 6.9 & 5.8 & 5.7 & 6.8 & 6.4 \\
& 120 & 6.4 & 5.6 & 6.2 & 6.2 & 5.8 & 6.0 & 5.9 & 6.4 & 5.2 \\
& 240 & 6.6 & 6.3 & 6.8 & 6.4 & 6.2 & 6.4 & 6.4 & 6.4 & 6.2 \\
\midrule
\multirow{3}{*}{$L_6$}
& 60  & 5.0 & 5.3 & 6.2 & 4.3 & 5.7 & 5.8 & 4.2 & 5.8 & 5.7 \\
& 120 & 4.6 & 4.5 & 5.9 & 4.2 & 5.1 & 6.0 & 5.0 & 5.3 & 6.1 \\
& 240 & 4.9 & 5.1 & 5.8 & 5.0 & 5.0 & 6.2 & 4.9 & 5.1 & 5.4 \\
\midrule
\multirow{3}{*}{$L_\infty$}
& 60  & 6.2 & 7.0 & 7.8 & 5.4 & 7.5 & 8.3 & 5.8 & 6.9 & 8.1 \\
& 120 & 4.4 & 5.7 & 6.4 & 4.7 & 5.5 & 7.1 & 4.6 & 5.6 & 6.9 \\
& 240 & 4.3 & 4.6 & 4.7 & 4.1 & 4.6 & 4.7 & 4.3 & 4.3 & 4.3 \\
\midrule
\multirow{3}{*}{minP}
& 60  & 7.0 & 7.5 & 8.2 & 6.8 & 8.1 & 7.2 & 6.6 & 8.3 & 8.3 \\
& 120 & 5.8 & 6.0 & 6.8 & 5.8 & 6.6 & 6.7 & 6.2 & 7.0 & 6.7 \\
& 240 & 5.4 & 6.6 & 5.8 & 5.7 & 6.7 & 5.8 & 5.9 & 6.1 & 6.2 \\
\midrule
\multirow{3}{*}{CC}
& 60  & 6.4 & 7.0 & 7.8 & 6.4 & 6.9 & 7.1 & 6.5 & 8.0 & 7.4 \\
& 120 & 6.2 & 6.1 & 6.7 & 5.9 & 6.4 & 6.9 & 6.4 & 6.7 & 6.7 \\
& 240 & 6.2 & 6.8 & 6.7 & 6.2 & 6.7 & 6.6 & 6.3 & 6.1 & 6.7 \\
\bottomrule
\end{tabular}%
\end{table}

\begin{table}[!htbp]
\centering
\caption{Empirical sizes (\%) of the six tests under normal errors and $\psi=0.25$.}
\label{tab:size_normal_6methods}
\setlength{\tabcolsep}{6pt}
\renewcommand{\arraystretch}{0.9}
\begin{tabular}{llccccccccc}
\toprule
& & \multicolumn{3}{c}{$\delta_{\gamma}=0$} 
  & \multicolumn{3}{c}{$\delta_{\gamma}=1/4$}
  & \multicolumn{3}{c}{$\delta_{\gamma}=1/2$} \\
\cmidrule(lr){3-5}\cmidrule(lr){6-8}\cmidrule(lr){9-11}
Method & $(T,N)$  & 50 & 100 & 200 & 50 & 100 & 200 & 50 & 100 & 200 \\
\midrule
\multirow{3}{*}{$L_2$}
& 60  & 7.9 & 6.1 & 6.6 & 7.6 & 6.6 & 6.8 & 7.8 & 7.1 & 7.5 \\
& 120 & 6.7 & 6.5 & 5.5 & 5.8 & 5.4 & 6.2 & 6.2 & 5.7 & 7.6 \\
& 240 & 5.6 & 5.9 & 5.5 & 6.1 & 6.3 & 5.9 & 7.3 & 5.7 & 6.3 \\
\midrule
\multirow{3}{*}{$L_4$}
& 60  & 7.6 & 6.2 & 8.1 & 7.5 & 6.4 & 7.8 & 7.1 & 7.8 & 7.0 \\
& 120 & 6.4 & 6.8 & 7.6 & 6.4 & 5.9 & 7.4 & 7.1 & 7.1 & 7.0 \\
& 240 & 6.8 & 6.5 & 5.6 & 7.2 & 5.9 & 5.8 & 7.4 & 6.0 & 6.9 \\
\midrule
\multirow{3}{*}{$L_6$}
& 60  & 6.3 & 5.4 & 7.5 & 6.4 & 5.5 & 6.7 & 5.6 & 6.7 & 5.9 \\
& 120 & 5.3 & 6.1 & 6.8 & 5.2 & 5.3 & 6.7 & 5.9 & 6.3 & 6.4 \\
& 240 & 5.6 & 5.7 & 6.1 & 6.2 & 5.1 & 6.0 & 5.8 & 5.4 & 6.0 \\
\midrule
\multirow{3}{*}{$L_\infty$}
& 60  & 7.3 & 7.4 & 10.1 & 6.8 & 7.8 & 8.8 & 7.0 & 8.4 & 7.6 \\
& 120 & 5.3 & 5.5 & 6.4 & 5.2 & 5.4 & 5.9 & 5.2 & 6.2 & 5.6 \\
& 240 & 5.1 & 4.9 & 4.4 & 5.1 & 4.4 & 4.9 & 5.1 & 4.9 & 4.7 \\
\midrule
\multirow{3}{*}{minP}
& 60  & 8.9 & 7.5 & 9.8 & 8.5 & 8.1 & 8.8 & 8.3 & 8.6 & 8.9 \\
& 120 & 7.0 & 6.9 & 7.0 & 6.2 & 6.4 & 6.8 & 7.1 & 6.8 & 7.1 \\
& 240 & 6.2 & 6.3 & 5.9 & 6.7 & 6.0 & 6.1 & 7.2 & 6.2 & 6.6 \\
\midrule
\multirow{3}{*}{CC}
& 60  & 8.6 & 6.9 & 9.4 & 8.3 & 7.4 & 8.7 & 7.7 & 8.5 & 8.2 \\
& 120 & 6.8 & 7.3 & 7.5 & 6.4 & 6.8 & 7.0 & 7.1 & 7.2 & 7.7 \\
& 240 & 6.6 & 7.1 & 6.2 & 7.5 & 6.2 & 6.5 & 7.6 & 6.2 & 7.0 \\
\bottomrule
\end{tabular}%
\end{table}

\begin{table}[!htbp]
\centering
\caption{Empirical sizes (\%) of the six tests under $t_8$ errors and $\psi=0.25$.}
\label{tab:size_t8_6methods}
\setlength{\tabcolsep}{6pt}
\renewcommand{\arraystretch}{0.9}
\begin{tabular}{llccccccccc}
\toprule
& & \multicolumn{3}{c}{$\delta_{\gamma}=0$} 
  & \multicolumn{3}{c}{$\delta_{\gamma}=1/4$}
  & \multicolumn{3}{c}{$\delta_{\gamma}=1/2$} \\
\cmidrule(lr){3-5}\cmidrule(lr){6-8}\cmidrule(lr){9-11}
Method & $(T,N)$ & 50 & 100 & 200 & 50 & 100 & 200 & 50 & 100 & 200 \\
\midrule

\multirow{3}{*}{$L_2$}
& 60  & 6.3 & 7.2 & 6.3 & 6.9 & 7.6 & 5.9 & 7.0 & 7.6 & 7.0 \\
& 120 & 6.8 & 6.2 & 5.9 & 6.7 & 6.7 & 6.0 & 6.9 & 7.8 & 5.9 \\
& 240 & 5.6 & 7.1 & 6.7 & 6.0 & 7.1 & 6.5 & 6.3 & 6.6 & 6.2 \\
\midrule

\multirow{3}{*}{$L_4$}
& 60  & 6.6 & 7.1 & 6.9 & 6.3 & 7.6 & 6.3 & 6.6 & 8.2 & 6.6 \\
& 120 & 6.4 & 6.4 & 6.3 & 6.7 & 6.5 & 6.7 & 6.0 & 6.8 & 6.1 \\
& 240 & 6.0 & 6.7 & 7.3 & 6.3 & 6.8 & 6.9 & 6.7 & 6.3 & 6.5 \\
\midrule

\multirow{3}{*}{$L_6$}
& 60  & 4.8 & 5.4 & 6.1 & 4.2 & 6.1 & 5.8 & 4.4 & 6.8 & 6.1 \\
& 120 & 5.3 & 5.1 & 5.6 & 5.1 & 5.4 & 6.3 & 4.8 & 5.7 & 5.6 \\
& 240 & 4.7 & 5.7 & 6.3 & 4.8 & 5.7 & 6.2 & 5.2 & 5.0 & 5.4 \\
\midrule

\multirow{3}{*}{$L_\infty$}
& 60  & 5.4 & 7.6 & 8.6 & 5.0 & 7.6 & 7.8 & 5.7 & 7.3 & 8.0 \\
& 120 & 4.9 & 5.3 & 6.2 & 4.4 & 5.1 & 6.2 & 4.6 & 5.1 & 6.2 \\
& 240 & 4.2 & 4.9 & 4.8 & 4.5 & 5.3 & 4.3 & 4.2 & 4.6 & 4.2 \\
\midrule

\multirow{3}{*}{minP}
& 60  & 7.3 & 8.9 & 8.8 & 7.0 & 9.4 & 8.0 & 7.4 & 9.8 & 9.2 \\
& 120 & 6.2 & 7.0 & 7.1 & 7.1 & 6.5 & 6.9 & 6.4 & 7.2 & 6.2 \\
& 240 & 5.7 & 6.9 & 6.6 & 5.9 & 6.8 & 5.9 & 6.2 & 6.6 & 6.2 \\
\midrule

\multirow{3}{*}{CC}
& 60  & 7.0 & 7.8 & 8.0 & 6.8 & 8.0 & 7.2 & 6.8 & 9.1 & 8.2 \\
& 120 & 6.4 & 6.8 & 7.0 & 6.9 & 7.0 & 7.1 & 6.3 & 6.9 & 6.6 \\
& 240 & 5.9 & 7.0 & 7.0 & 6.4 & 7.0 & 6.6 & 6.6 & 6.0 & 6.6 \\
\bottomrule
\end{tabular}%
\end{table}

For the power study, we consider alternatives for
$\bal=(\alpha_1,\ldots,\alpha_N)^{\T}$ with exactly
$n$ nonzero components. Specifically, letting
$S_\alpha\subset\{1,\ldots,N\}$ denote the support set with
$|S_\alpha|=n$, we generate
\[
\alpha_i=\frac{4}{n^{1/2.2}}\,s_i,
\quad i\in S_\alpha,
\qquad
\alpha_i=0,
\quad i\notin S_\alpha.
\]
 In the baseline
setting, we take $S_\alpha=\{1,\ldots,n\}$ and use alternating
signs $s_i\in\{1,-1\}$ over the active coordinates. Thus, as
$n$ increases, the individual nonzero signals become weaker,
allowing us to assess the power of different tests against sparse alternatives
with varying degrees of sparsity. 

Here we fix $(T,N)=(120,200)$ and consider three levels of latent factor strength, $\delta_{\gamma}\in\{0,1/4,1/2\}$. The corresponding power curves are displayed in Figures~\ref{fig1}--\ref{fig3}, where suffixes $1$, $2$, and $3$ represent $\delta_{\gamma}=0$, $1/4$, and $1/2$, respectively.

Overall, the three figures show a very similar ordering of the six procedures across the different values of $\delta_{\gamma}$. When the alternative is extremely sparse, namely when only a very small number of $\alpha_i$'s are nonzero, the higher-order and max-type procedures are clearly more effective than the sum-type statistic. In particular, $L_{\infty}$ and $L_6$ deliver the highest power at the very sparse end, while minP and CC also perform very competitively and remain close to the best individual method. As the number of nonzero signals increases, the advantage gradually shifts away from the extreme-sparsity-oriented procedures. In the moderately sparse region, $L_4$ becomes particularly attractive: it has consistently strong power, often close to the best method, and provides a simple and stable choice when one prefers not to rely on combining multiple tests. In this sense, $L_4$ offers a good compromise between sensitivity to sparse signals and robustness across a broader range of alternatives.

A more refined comparison between the two combination methods shows that CC is generally preferable to minP in the moderately sparse regime. Although the two procedures are very close overall, CC tends to have a visible advantage when the signal is neither extremely sparse nor very dense, suggesting that combining information from all component tests is beneficial in this intermediate region. By contrast, when the alternative becomes quite dense, minP is typically slightly stronger than CC, although the gap is not large. Thus, the relative comparison between CC and minP depends on the sparsity level: CC is more favorable in moderate sparsity, whereas minP can have a small edge once the signal becomes sufficiently diffuse.

For larger values of $n$, the $L_2$ statistic eventually becomes the most powerful procedure. This is expected, because under the present design each nonzero intercept is scaled by $n^{-1/2.2}$, so the alternative becomes less sparse but also weaker componentwise as $n$ increases. Such alternatives are more favorable to sum-type aggregation, and hence $L_2$ dominates in the denser regime. The impact of increasing $\delta_{\gamma}$ from $0$ to $1/4$ and then to $1/2$ is relatively mild: all methods experience only a modest loss of power, and the overall ranking of the procedures remains essentially unchanged. Taken together, the results suggest that $L_{\infty}$ and $L_6$ are most effective for very sparse alternatives, CC is especially competitive in the moderately sparse region, minP becomes slightly preferable to CC in the denser region, and $L_4$ serves as a particularly appealing simple choice when a single non-combined procedure is desired.

\begin{figure}[htb]
	\centering
	\includegraphics[width=\linewidth]{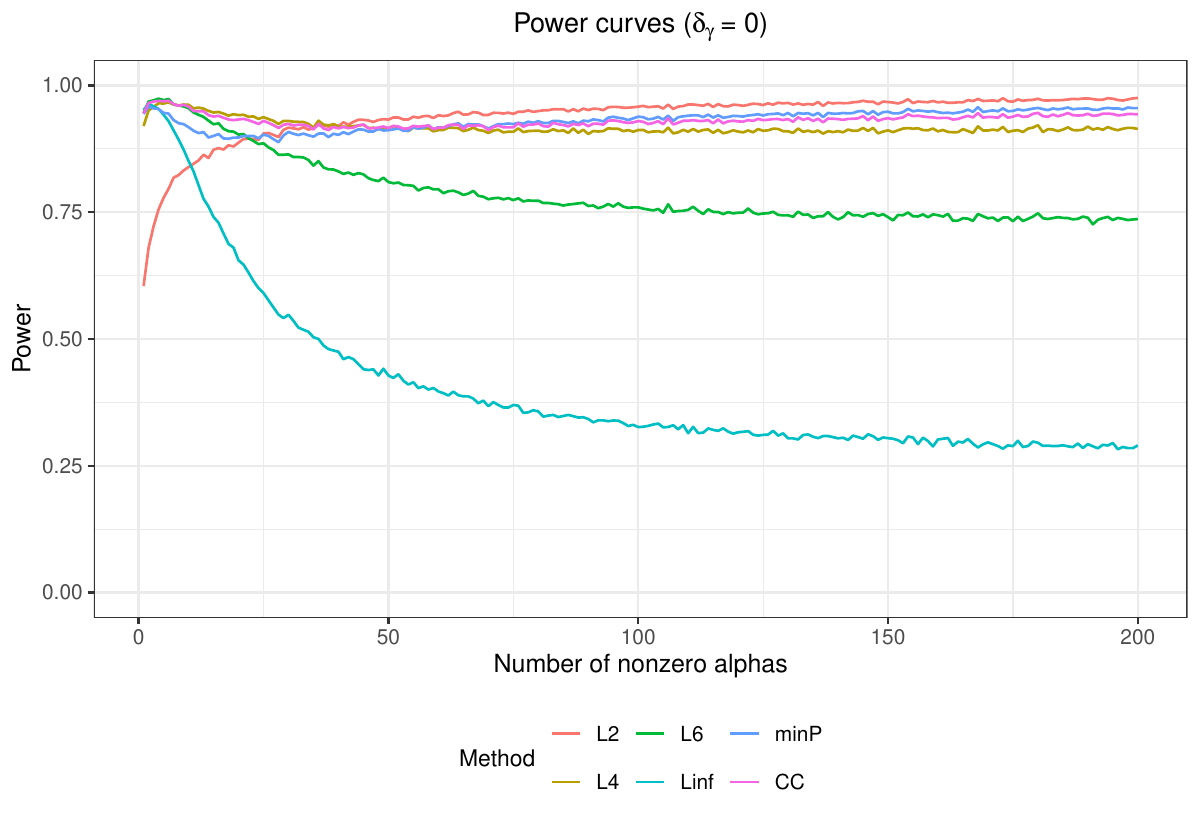}
 \caption{Power curves of six methods with $\delta_\gamma=0$ and $(T,N)=(120,200)$.}\label{fig1}
\end{figure}

\begin{figure}[htb]
	\centering
	\includegraphics[width=\linewidth]{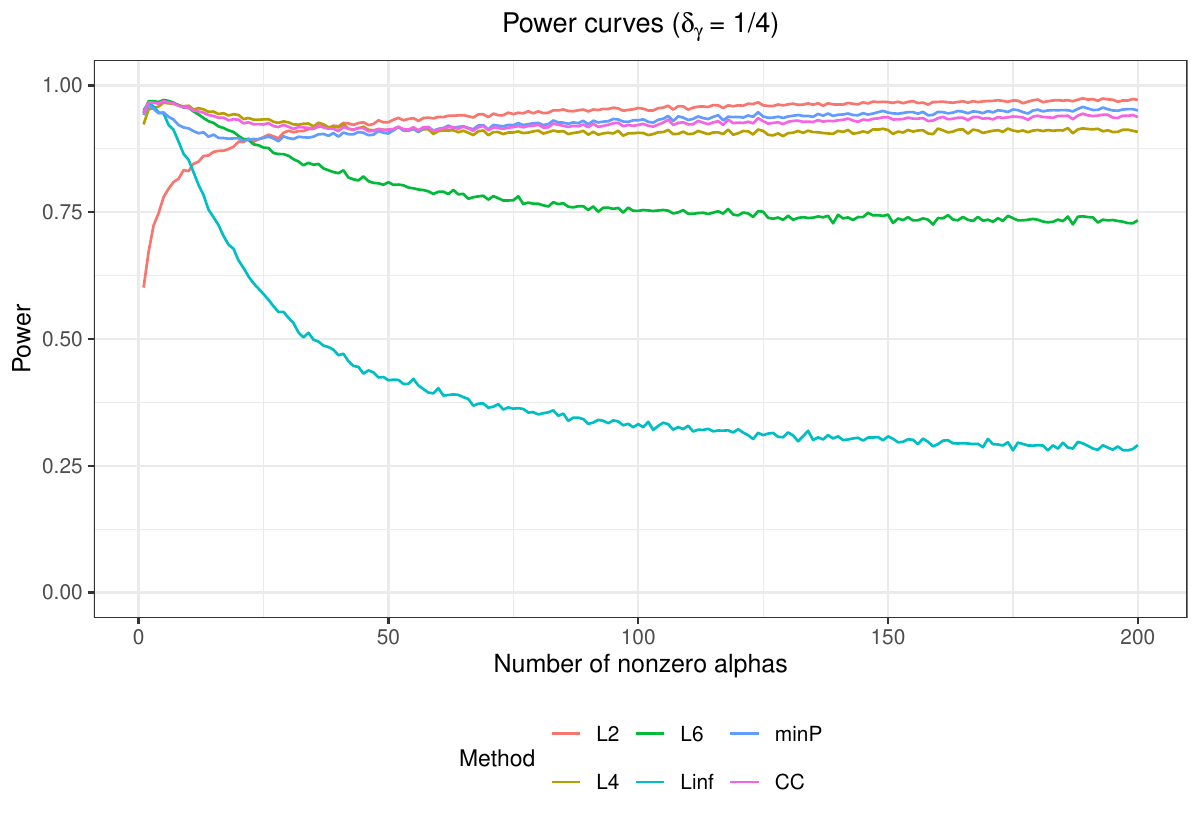}
 \caption{Power curves of six methods with $\delta_\gamma=1/4$ and $(T,N)=(120,200)$.}\label{fig2}
\end{figure}

\begin{figure}[htb]
	\centering
	\includegraphics[width=\linewidth]{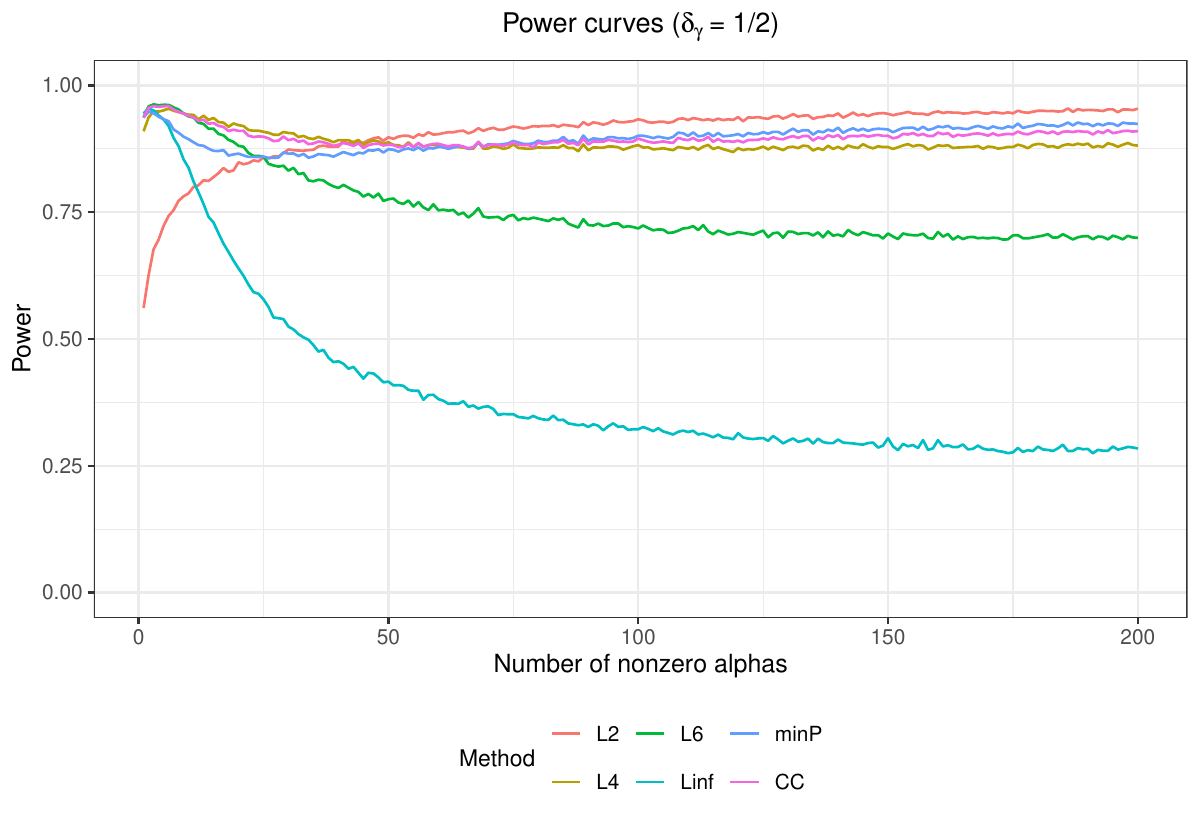}
 \caption{Power curves of six methods with $\delta_\gamma=1/2$ and $(T,N)=(120,200)$.}\label{fig3}
\end{figure}

\section{Real Data Application}
\label{sec:empirical_capm}

In this section, we examine whether the CAPM can account for the cross section of U.S. equity returns over time. Our analysis is based on monthly stock price data for S\&P 500 constituents from January 2005 to January 2024. Monthly stock returns are computed from month-end prices, and the factor data are matched with the corresponding monthly observations from the Fama--French data set. Since the CAPM is nested in the Fama--French factor structure, we use only the market excess return, \( \text{Mkt}-\text{RF} \), together with the risk-free rate \( \text{RF} \), and form the excess return of asset \(i\) in month \(t\) as
\[
Y_{it}=R_{it}-RF_t.
\]
For each asset \(i\), the CAPM over a given subsample is
\[
Y_{it}=\alpha_i+\beta_i(\text{Mkt}-\text{RF})_t+\varepsilon_{it},
\qquad t=1,\dots,T,
\]
and the null hypothesis of interest is
\[
H_0:\alpha_1=\cdots=\alpha_N=0.
\]

To investigate the time variation in pricing performance, we adopt a rolling-window design. For a given window length \(T\), we move the window forward one month at a time and apply the six tests developed in this paper to each subsample. These are the \(L_2\), \(L_4\), \(L_6\), and \(L_\infty\) tests, together with the two combination procedures minP and CC. The implementation follows exactly the algorithm used in our simulation study. In each window, we first keep only those securities with complete monthly return observations over the entire window, so the cross-sectional dimension is allowed to vary over time. We proceed only when at least 100 securities remain in the window. We then estimate the CAPM for each retained security by ordinary least squares, compute the studentized alpha statistics, estimate the residual dependence structure using the same thresholding rule as in the baseline implementation, and finally obtain the six \(p\)-values. This procedure yields a sequence of rolling-window \(p\)-values for each method, which allows us to assess both the timing and the frequency of CAPM rejections.

Figure~\ref{fig:capm_pvalue_curves} plots the resulting rolling \(p\)-value paths, while Table~\ref{tab:capm_rej_window} summarizes the rejection frequencies at the 5\% level for window lengths \(T=12\times(5{:}10)\), namely \(T=60,72,84,96,108,\) and \(120\). 

Several conclusions emerge from the CAPM application. First, the overall evidence points to systematic departures from the zero-alpha restriction over a substantial part of the sample period. Among the individual \(L_q\)-type procedures, the \(L_2\) and \(L_4\) tests are the most informative and stable, whereas the \(L_\infty\) test is uniformly less sensitive. This ranking is consistent with our simulation findings. In particular, the relatively weak performance of \(L_\infty\) suggests that the empirical departures from the CAPM are not driven primarily by a very small number of extreme alphas. Instead, the stronger performance of \(L_2\), \(L_4\), and the two combination procedures indicates a more moderately sparse or diffuse pattern of mispricing across the cross section. The behavior of \(L_6\) lies between these two extremes, which is also broadly in line with the simulation evidence that higher-order norms become more competitive when the signal is stronger but may be less stable than lower-order norms in moderate regimes.

Second, the comparison across methods further highlights the value of combination-based inference. The minP and CC procedures remain highly competitive throughout, and their performance is close to, or slightly better than, that of the best individual \(L_q\)-based tests over a range of window lengths. This pattern accords with the motivation of the combined procedures: when the underlying sparsity level is unknown, combining information from several norms offers a more robust strategy than relying on any single test alone.

Third, the rolling-window \(p\)-value plots provide useful time-series evidence on when the CAPM fails. For \(L_2\), \(L_4\), minP, and CC, the \(p\)-values fall below the 5\% threshold over an extended interval rather than only at isolated dates, indicating that the rejection of the CAPM is persistent rather than episodic. Moreover, the low-\(p\)-value region broadly covers the COVID-19 period, especially for the intermediate and longer window lengths, which suggests that the CAPM has difficulty explaining stock returns during that episode of severe market disruption. At the same time, the rejection region is not confined to the pandemic alone, implying that the model misspecification is more structural than event-specific. Overall, these empirical findings reinforce the message from the theory and simulations: lower- and moderate-order \(L_q\) tests, together with their combinations, provide the most reliable evidence against the null when the alternative is neither extremely sparse nor fully dense.

\begin{figure}[htb]
	\centering
	\includegraphics[width=\linewidth]{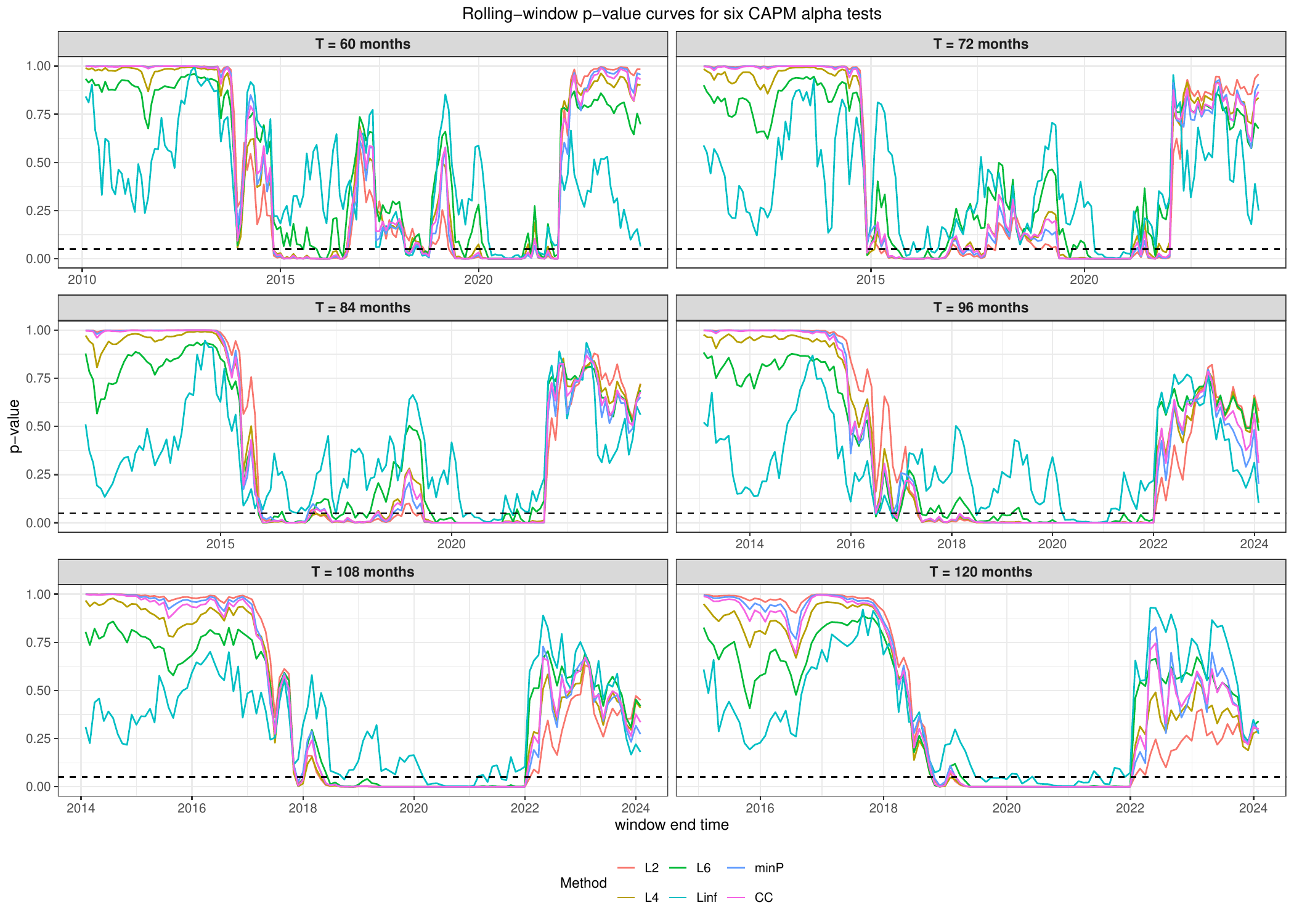}
 \caption{p-value paths of the six tests over rolling windows.}\label{fig:capm_pvalue_curves}
\end{figure}

\begin{table}[!htbp]
\centering
\caption{Rejection rates (\%) of the six tests under the rolling-window CAPM analysis for different window lengths \(T\).}
\label{tab:capm_rej_window}
\setlength{\tabcolsep}{8pt}
\renewcommand{\arraystretch}{1.1}
\begin{tabular}{lcccccc}
\toprule
Method & \(T=60\) & \(T=72\) & \(T=84\) & \(T=96\) & \(T=108\) & \(T=120\) \\
\midrule
\(L_2\)                & 33.73 & 38.85 & 46.90 & 42.86 & 39.67 & 34.86 \\
\(L_4\)                & 33.14 & 33.12 & 45.52 & 42.86 & 38.84 & 35.78 \\
\(L_6\)                & 18.93 & 19.11 & 31.03 & 36.84 & 37.19 & 33.03 \\
\(L_\infty\)           & 13.02 & 11.46 &  8.28 & 10.53 & 12.40 & 19.27 \\
minP                   & 36.09 & 35.67 & 42.07 & 42.86 & 37.19 & 33.94 \\
CC                     & 34.91 & 33.76 & 42.07 & 43.61 & 38.02 & 33.94 \\
\bottomrule
\end{tabular}
\end{table}

\section{Conclusion}
\label{sec:conc}

This paper studies testing alpha in high-dimensional linear factor pricing models. We develop a family of \(L_q\)-type procedures, including the \(L_2\), \(L_4\), \(L_6\), and \(L_\infty\) tests, establish their asymptotic null behavior, and show how their power depends on the sparsity of the alternative. We further propose an adaptive Cauchy combination test, which integrates information from different \(L_q\)-based procedures and performs well across a broad range of alternatives. Overall, the paper provides a unified framework for high-dimensional alpha testing under linear factor pricing models.

There are several important directions for future research. First, the current analysis is developed under temporally independent observations. In many financial applications, however, asset returns exhibit serial dependence and more general dynamic dependence structures. It would therefore be of substantial interest to extend the proposed \(L_q\)-based and adaptive testing framework to dependent observations; see, for example, \citet{MaFengWangBao2024Dep}. Second, the present paper focuses on a static linear factor pricing model. In many empirical applications, factor loadings and risk premia may vary over time with conditioning information or state variables. Extending the present high-dimensional alpha testing framework to conditional time-varying factor models is therefore a natural and important topic for future research; see \citet{FersonHarvey1991,JagannathanWang1996,LettauLudvigson2001,MaLanSuTsai2020,MaFengWangBao2024}. We leave these problems for future research.

\section{Appendix}

\subsection{Useful Lemmas}
\label{app:prelim}

This appendix collects the simple bounds that drive the remainder calculations. Recall \(\b=(b_1,\ldots,b_T)^{\T}:={\rm M}_{\F}\bone_T/\sqrt{w_T}\). Under Assumption~\ref{ass:factor}, \(\sum_t b_t^2=1\) and \(\max_t|b_t|\le C_bT^{-1/2}\). For every component index \(k\), define
\[
s_k=\sum_{t=1}^T b_t z_{kt}.
\]
Under Assumption \ref{ass:ICM}, the variables \(s_1,\dots,s_N\) are independent and identically distributed over \(k\).

\begin{lemma}[Powers of the weight vector]
\label{lem:bpower}
Under Assumption \ref{ass:factor}, for every integer \(r\ge 2\),
\[
\sum_{t=1}^T|b_t|^r\le C_rT^{1-r/2}.
\]
\end{lemma}

\begin{proof}
Using \(\sum_t b_t^2=1\) and \(\max_t|b_t|\le C_bT^{-1/2}\),
\[
\sum_{t=1}^T|b_t|^r
=
\sum_{t=1}^T|b_t|^{r-2}b_t^2
\le
\Bigl(\max_{1\le t\le T}|b_t|\Bigr)^{r-2}\sum_{t=1}^Tb_t^2
\le
C_b^{r-2}T^{-(r-2)/2}.
\]
Since \(-(r-2)/2=1-r/2\), the claim follows.
\end{proof}

\begin{lemma}[Cumulants of the projected components]
\label{lem:skcumulant}
Let \(\kappa_r(\cdot)\) denote the ordinary cumulant of order \(r\). Then under Assumption~\ref{ass:factor} and \ref{ass:ICM}, for every integer \(r\ge3\),
\[
|\kappa_r(s_k)|\le C_rT^{1-r/2}.
\]
\end{lemma}

\begin{proof}
By independence over \(t\), cumulants add:
\[
\kappa_r(s_k)
=
\kappa_r\Bigl(\sum_{t=1}^Tb_tz_{kt}\Bigr)
=
\sum_{t=1}^Tb_t^r\kappa_r(z_{kt}).
\]
Because \(z_{kt}\) has the same law for all \((k,t)\),
\[
|\kappa_r(s_k)|
\le |
\kappa_r(z_{11})|
\sum_{t=1}^T|b_t|^r.
\]
Lemma~\ref{lem:bpower} gives the stated order.
\end{proof}

The next lemma explains why the assumption \(E(z_{it}^3)=0\) is unnecessary for the present paper. All the statistics studied in the paper are built from \emph{even} powers. As a result, odd cumulants can only appear in pairs, and therefore their contribution is still of order at most \(T^{-1}\).

\begin{lemma}[Why the third moment assumption can be removed]
\label{lem:no-third-moment}
Fix a collection of nonnegative integers \((\ell_1,\dots,\ell_q)\) such that
\[
\ell_1+\cdots+\ell_q\le 12
\qquad\text{and}\qquad
\ell_1+\cdots+\ell_q\ \text{is even}.
\]
Let \(g_1,\dots,g_q\) be i.i.d.\ standard normals. Then
\[
E\Bigl(\prod_{j=1}^q s_{k_j}^{\ell_j}\Bigr)
=
E\Bigl(\prod_{j=1}^q g_j^{\ell_j}\Bigr)+O(T^{-1}),
\]
whenever \(k_1,\dots,k_q\) are distinct. If the total degree is odd, then
\[
E\Bigl(\prod_{j=1}^q s_{k_j}^{\ell_j}\Bigr)=O(T^{-1/2}).
\]
\end{lemma}

\begin{proof}
Apply the moment--cumulant formula. The Gaussian moment is obtained by summing over pair partitions only. Every non-Gaussian correction therefore contains at least one block of size other than two.

Suppose the total degree is even. Then a partition cannot contain exactly one odd block, because the sum of the remaining block sizes would then be odd. Therefore every partition that contains an odd block must contain at least two odd blocks. By Lemma~\ref{lem:skcumulant}, any odd block has size at least three and contributes a factor of order at most \(T^{-1/2}\). Two such blocks therefore contribute at most \(T^{-1}\). If a partition contains no odd block but contains some block of size at least four, then Lemma~\ref{lem:skcumulant} gives an \(O(T^{-1})\) factor directly. Thus every non-pairing partition contributes at most \(CT^{-1}\). Because the total degree is at most twelve, the number of partitions is a fixed constant, so the full remainder is \(O(T^{-1})\).

If the total degree is odd, every admissible partition contains at least one odd block, and hence at least one factor of order \(T^{-1/2}\). This proves the second display.
\end{proof}

Recall \(\D=\diag(\sigma_1^2,\dots,\sigma_N^2)\) and \(\R=\D^{-1/2}\bmS \D^{-1/2}\). Let \(\A=(a_{ij})_{1\le i,j\le N}=\D^{-1/2}\bmS^{1/2}\). Then \(\A\A^{\T}=\R\).

\begin{lemma}[Oracle representation]
\label{lem:oraclerep}
Under Assumption~\ref{ass:ICM} and \(H_0\),
\[
\x=(x_1,\dots,x_N)^{\T}=\A\s,
\qquad \s=(s_1,\dots,s_N)^{\T}.
\]
Consequently,
\[
E(x_i)=0,
\qquad
E(x_ix_j)=r_{ij}.
\]
\end{lemma}

\begin{proof}
By the definition of \(x_i\),
\[
x_i=\frac{\sqrt{w_T}\widehat\alpha_i}{\sigma_i}
=
\frac{1}{\sigma_i\sqrt{w_T}}\bone_T^{\T}{\rm M}_{\F}\bme_{i.}
=
\sum_{t=1}^Tb_t\frac{\varepsilon_{it}}{\sigma_i}.
\]
Stacking the \(N\) coordinates and using \(\bme_t=\bmS^{1/2}\z_t\),
\[
\x=\D^{-1/2}\sum_{t=1}^Tb_t\bme_t
=\D^{-1/2}\bmS^{1/2}\sum_{t=1}^Tb_t\z_t
=\A\s.
\]
Then
\[
\E(\x\x^{\T})=\A\E(\s\s^{\T})\A^{\T}=\A\A^{\T}=\R,
\]
because the \(s_k\)'s are independent with unit variance.
\end{proof}

\begin{lemma}[Mixed moments of the oracle vector]
\label{lem:oracle-mixed}
Let \(\y=(y_1,\dots,y_N)^{\T}=\A\g\) with $\g=(g_1,\dots,g_N)^{\T}\sim \mathcal{N}(\bzero,\I_N)$. Under Assumption~\ref{ass:factor}-\ref{ass:dep} and \(H_0\), for even integers \(a,b\ge 0\) with and \(a+b\le 12\),
\[
N^{-1}\sum_{1\le i,j\le N}\left\{\E(x_i^ax_j^b)-\E(y_i^ay_j^b)\right\}=O(T^{-1}).
\]
\end{lemma}
\begin{proof}
    Expand \(x_i\) via Lemma~\ref{lem:oraclerep} that \(x_i=\sum_{k=1}^Na_{ik}s_k\). Using $\E(s_k)=\E(g_k)=0$, 
    \begin{align*}
        &\E(x_i^ax_j^b)-\E(y_i^ay_j^b)\\
        =&\E\left\{\left(\sum_{k=1}^Na_{ik}s_k\right)^a\left(\sum_{k=1}^Na_{gk}s_k\right)^b-\left(\sum_{k=1}^Na_{ik}g_k\right)^a\left(\sum_{k=1}^Na_{gk}g_k\right)^b\right\}\\
        =&\sum_{1\le k_1,\dots,k_a,k_{a+1},\dots,k_{a+b}\le N}a_{ik_1}\dots a_{ik_a}a_{jk_{a+1}}\dots a_{jk_{a+b}}\E(s_{k_1}\dots s_{k_{a+b}}-g_{k_1}\dots g_{k_{a+b}})\\
        =&\sum_{\substack{1\le k_1,\dots,k_a,k_{a+1},\dots,k_{a+b}\le N\\\#\text{distinct}\le \frac{a+b}{2}}}a_{ik_1}\dots a_{ik_a}a_{jk_{a+1}}\dots a_{jk_{a+b}}\E(s_{k_1}\dots s_{k_{a+b}}-g_{k_1}\dots g_{k_{a+b}}).
    \end{align*}
    Lemma~\ref{lem:no-third-moment} implies \(\E(s_{k_1}\dots s_{k_{a+b}}-g_{k_1}\dots g_{k_{a+b}})=O(T^{-1})\). For the coefficients, observing that \(AA^T=R\), then under the weak cross-sectional dependence Assumption~\ref{ass:dep}, 
    \begin{align*}
        N^{-1}\sum_{i,j}\sum_{\substack{1\le k_1,\dots,k_a,k_{a+1},\dots,k_{a+b}\le N\\\#\text{distinct}\le \frac{a+b}{2}}}a_{ik_1}\dots a_{ik_a}a_{jk_{a+1}}\dots a_{jk_{a+b}}\lesssim CN^{-1}\sum_{i,j}r_{ij}^{\frac{a+b}{2}}=O(1).
    \end{align*}
    The required bound follows.
\end{proof}

\begin{proposition}[Oracle covariance matrix]
\label{prop:oraclecov}
Let \(\Q_N^{\cir}=(Q_{2,N}^{\cir},Q_{4,N}^{\cir},Q_{6,N}^{\cir})^{\T}\), and let \(\B_N^{(0)}\) denote the matrix obtained from \(\B_N\) by deleting the \(v^{-1}\) correction terms in the diagonal entries. Then under Assumption~\ref{ass:factor}-\ref{ass:dep} and \(H_0\),
\[
\Cov(\Q_N^{\cir})=\B_N^{(0)}+\bD_N,
\qquad
\max_{1\le k,\ell\le 3}|(\Delta_N)_{k\ell}|\le CT^{-1}.
\]
Moreover,
\[
\max_{1\le k,\ell\le 3}|B_{k\ell,N}-B^{(0)}_{k\ell,N}|
\le C v^{-1}.
\]
\end{proposition}

\begin{proof}
Every covariance entry is an average of terms of the form
\[
E(x_i^ax_j^b)-E(x_i^a)E(x_j^b),
\qquad a+b\le 12,
\]
with \((a,b)\in\{(2,2),(2,4),(2,6),(4,4),(4,6),(6,6)\}\). By Lemma~\ref{lem:oracle-mixed}, we may replace \((x_i,x_j)\) by the Gaussian pair \((y_i,y_j)\) at an \(O(T^{-1})\) cost.

For jointly Gaussian \((Y_1,Y_2)\) with correlation \(\rho\), we use the algebraic identities
\[
Y^2-1=H_2(Y),
\qquad
Y^4-3=H_4(Y)+6H_2(Y),
\qquad
Y^6-15=H_6(Y)+15H_4(Y)+45H_2(Y),
\]
and the well-known formula
\[
\E\{H_m(Y_1)H_n(Y_2)\}=\ind(m=n)\,m!\rho^m.
\]
This gives
\begin{align*}
\Cov(Y_1^2-1,Y_2^2-1)&=2\rho^2,\\
\Cov(Y_1^2-1,Y_2^4-3)&=12\rho^2,\\
\Cov(Y_1^2-1,Y_2^6-15)&=90\rho^2,\\
\Cov(Y_1^4-3,Y_2^4-3)&=72\rho^2+24\rho^4,\\
\Cov(Y_1^4-3,Y_2^6-15)&=540\rho^2+360\rho^4,\\
\Cov(Y_1^6-15,Y_2^6-15)&=4050\rho^2+5400\rho^4+720\rho^6.
\end{align*}
After averaging over \(\rho=r_{ij}\), we obtain the entries of \(B_N^{(0)}\). The replacement error is \(O(T^{-1})\) uniformly in each entry.

It remains to compare \(B_N^{(0)}\) with the corrected benchmark \(B_N\). By construction, the two matrices differ only in the diagonal entries, with
\begin{align*}
B_{22,N}-B_{22,N}^{(0)}
&=
\frac{1}{vN}\sum_{i,j=1}^N \bigl(10r_{ij}^2+4r_{ij}^4\bigr),\\
B_{44,N}-B_{44,N}^{(0)}
&=
\frac{1}{vN}\sum_{i,j=1}^N \bigl(936r_{ij}^2+864r_{ij}^4+192r_{ij}^6\bigr),\\
B_{66,N}-B_{66,N}^{(0)}
&=
\frac{1}{vN}\sum_{i,j=1}^N \bigl(101250r_{ij}^2+202500r_{ij}^4+114480r_{ij}^6+12960r_{ij}^8\bigr).
\end{align*}
Assumption~\ref{ass:dep} implies that the averages of \(r_{ij}^{2k}\), \(k=1,2,3,4\), are uniformly bounded. Hence
\[
\max_{1\le k,\ell\le3}|B_{k\ell,N}-B^{(0)}_{k\ell,N}|
\le Cv^{-1}.
\]
The additional \(O(v^{-2})\) remainder in the exact Gaussian diagonal variances from Section~2.1 is uniformly negligible because \(v\asymp T\to\infty\). This proves the claim.
\end{proof}

\begin{lemma}
\label{lem:sigmaexpand}
Under Assumption~\ref{ass:factor}-\ref{ass:dep} and \(H_0\), for each \(i\),
\[
\frac{\widehat\sigma_i^2}{\sigma_i^2}-1
=
L_{iT}+R_{iT},
\]
where
\[
L_{iT}=
\sum_{1\le s\ne t\le T}c_{st}\,\eta_{is}\eta_{it},
\qquad \eta_{it}=\frac{\varepsilon_{it}}{\sigma_i}, \qquad R_{iT}=\sum_{t=1}^Tc_{tt}(\eta_{it}^2-1),
\]
for deterministic coefficients satisfying
\[
|c_{st}|\le CT^{-1},
\qquad
\sum_{s,t=1}^T c_{st}^2\le CT^{-1},
\qquad
\sum_{t=1}^T|c_{tt}|\le C,
\]
and
\[
\E\left(\left|\frac{\widehat\sigma_i^2}{\sigma_i^2}-1\right|^2\right)=O(T^{-1}).
\]
\end{lemma}

\begin{proof}
Because the number of factors is fixed, the OLS residual variance is a quadratic form in \(\bme_{i.}=(\varepsilon_{i1},\dots,\varepsilon_{iT})^{\T}\):
\[
\widehat\sigma_i^2
=
\frac1v\bme_{i.}^{\T}{\rm M}_{\X}\bme_{i.},
\]
where \({\rm M}_{\X}=\I_T-\X(\X^{\T}\X)^{-1}\X^{\T}\) with $\X=(\bone_T,\F)$, i.e., the residual-maker with respect to the intercept and the factors. Write
\[
v^{-1}{\rm M}_{\X}=(c_{st})_{1\le s,t\le T}.
\]
Since \({\rm M}_{\X}\) is idempotent of rank \(v\asymp T\), its operator norm is one and \(\tr({\rm M}_{\X})=\tr({\rm M}_{\X}^2)=v\). Therefore each \(c_{st}\) is \(O(T^{-1})\), \(\sum_{t}c_{tt}=v^{-1}\tr({\rm M}_{\X})=1\) and \(\sum_{s,t}c_{st}^2=v^{-2}\tr({\rm M}_{\X}^2)=v^{-1}\).

Let $\eta_{it}=\varepsilon_{it}/\sigma_i$ and $\bet_{i.}=(\eta_{i1},\dots,\eta_{iT})^{\T}$, then we have the following decomposition:
\[\frac{\widehat\sigma_i^2}{\sigma_i^2}-1=\bet_{i.}^{\T}(v^{-1}{\rm M}_{\X})\bet_{i.}-1=L_{iT}+R_{iT},\]
where
\[
L_{iT}=
\sum_{1\le s\ne t\le T}c_{st}\,\eta_{is}\eta_{it}, \qquad R_{iT}=\sum_{t=1}^Tc_{tt}(\eta_{it}^2-1).
\]

Note that $\eta_{it}=\sigma_i^{-1}\sum_{k=1}^N(\bmS^{1/2})_{ik}z_{kt}$ are independent over $t$ with $\E(\eta_{it})=0$, $\E(\eta_{it}^2)=1$ and $\max_{i,t}\Vert\eta_{it}\Vert_{\psi_2}=O(1)$, thus
\[
\E(|L_{iT}+R_{iT}|^2)\le 2\{\E(|L_{iT}|^2)+\E(|R_{iT}|^2)\}\le C\sum_{s,t}c_{st}^2=O(T^{-1}).
\]

\end{proof}

\begin{lemma}
\label{lem:keymixed}
Under Assumption~\ref{ass:factor}-\ref{ass:growth} and \(H_0\), for each \(a\in\{2,4,6\}\),
\[
N^{-1/2}\sum_{i=1}^Nx_i^a\Bigl(\frac{\widehat\sigma_i^2}{\sigma_i^2}-1\Bigr)=o_p(1).
\]
\end{lemma}

\begin{proof}
    Recall that $\bet_t=\D^{-1/2}\bme_t=\A\z_t$, 
    \[x_i^2=\bet_{i.}^{\T}(w_T^{-1}{\rm H}_{\F})\bet_{i.},\qquad X_i:=\hat{\sigma}_i^2/\sigma_i^2=\bet_{i.}^{\T}(v^{-1}{\rm M}_{\X})\bet_{i.}\]
    where ${\rm H}_{\F}={\rm M}_{\F}\bone_T\bone_T^{\T}{\rm M}_{\F}$. For each $i$, $\bet_{i.}=(\eta_{i1},\dots,\eta_{iT})^{\T}\sim IID(\bzero,\I_T)$, and $\Vert\eta_{it}\Vert_{\psi_2}=O(1)$ as $\Vert z_{it}\Vert_{\psi_2}=O(1)$. Next, using the moments of products of quadratic forms under non-Gaussianity (for example, see Lemma 6 in \cite{PesaranYamagata2024}), we analyze the order of $N^{-1/2}\sum_{i}x_i^a(X_i-1)$.

    For $a=2$, 
    \begin{align*}
        \E(x_i^2)=&\E\{\bet_{i.}^{\T}(w_T^{-1}{\rm H}_{\F})\bet_{i.}\}=\tr\{(w_T^{-1}{\rm H}_{\F})\}=1,\\
        \E(x_i^2X_i)=&\E\{\bet_{i.}^{\T}(w_T^{-1}{\rm H}_{\F})\bet_{i.}\bet_{i.}^{\T}(v^{-1}{\rm M}_{\X})\bet_{i.}\}\\
        =&\gamma_{\eta,2}\tr\{(w_T^{-1}{\rm H}_{\F})\odot(v^{-1}{\rm M}_{\X})\}+\tr\{(w_T^{-1}{\rm H}_{\F})\}\tr\{(v^{-1}{\rm M}_{\X})\}\\
        &+2\tr\{(w_T^{-1}{\rm H}_{\F})(v^{-1}{\rm M}_{\X})\}\\
        =&1+O(T^{-1}),
    \end{align*}
    where $\gamma_{\eta,2}:=\E(\eta_{it}^4)-3=O(1)$. Here, the last step follows from the fact that ${\rm M}_{\X}{\rm H}_{\F}=\bzero$ and $\tr({\rm H}_{\F}\odot{\rm M}_{\X})=O(T)$. Hence,
    \[
    N^{-1/2}\sum_i\E\{x_i^2(X_1-1)\}=O(N^{1/2}T^{-1}).
    \]
    Furthermore,
    \begin{align*}
        &\var\left\{N^{-1/2}\sum_ix_i^2(X_i-1)\right\}\\
        =&N^{-1}\sum_i\var\{x_i^2(X_i-1)\}+N^{-1}\sum_{i\ne j}\Cov\{x_i^2(X_i-1),x_j^2(X_j-1)\}
    \end{align*}
    We first calculate that \(\E(x_i^4)=3+O(T^{-1})\), \(\E(x_i^4X_i)=3+O(T^{-1})\) and \(\E(x_i^4X_i^2)=3+O(T^{-1})\) (for details, see Lemma 11 in \cite{PesaranYamagata2024}). Thus,
    \[
    \var\{x_i^2(X_i-1)\}=\E\{x_i^4(X_i-1)^2\}-[\E\{x_i^2(X_i-1)\}]^2=O(T^{-1})+O(T^{-2})=O(T^{-1}).
    \]
    Next, we prove (for details, see Lemma 16 in \cite{PesaranYamagata2024}) 
    \[
    N^{-1}\sum_{i\ne j}\Cov\{x_i^2(X_i-1),x_j^2(X_j-1)\}=O(T^{-1}+NT^{-2}).
    \]
    In summary, 
    \[
    N^{-1/2}\sum_ix_i^2(X_i-1)=O_p(T^{-1/2}+N^{1/2}T^{-1})=o_p(1),
    \]
    as long as $N^{1/2}T^{-1}\to 0$ and $\min(N,T)\to\infty$.

    Using higher moments of quadratic forms, we can prove that \(N^{-1/2}\sum_ix_i^a(X_i-1)=o_p(1)\), for $a=4,6$.
\end{proof}

\begin{proposition}
\label{prop:oracle-rate}
Under Assumptions~\ref{ass:factor}--\ref{ass:growth}, for each \(a\in\{2,4,6\}\),
\[
Q_{a,N}-Q_{a,N}^{\cir}=o_p(1).
\]
\end{proposition}

\begin{proof}
Let
\[
X_i=\frac{\widehat\sigma_i^2}{\sigma_i^2}.
\]
Then \(t_i^a=x_i^aX_i^{-a/2}\), so
\[
Q_{a,N}-Q_{a,N}^{\cir}
=
\frac1{\sqrt N}\sum_{i=1}^N x_i^a\{X_i^{-a/2}-1\}
+
\sqrt N(\mu_a-\mu_{a,v}).
\]
Because \(v=T-p-1\asymp T\), the Student centering correction satisfies
\[
\sqrt N|\mu_a-\mu_{a,v}|\le C\frac{\sqrt N}{T}.
\]
It remains to study the first term.

On the event \(\max_i|X_i-1|\le1/2\), a Taylor expansion gives
\[
X_i^{-a/2}-1
=
-\frac a2(X_i-1)+\frac{a(a+2)}8(X_i-1)^2+\rho_{ai},
\qquad |\rho_{ai}|\le C|X_i-1|^3.
\]
Hence
\[
Q_{a,N}-Q_{a,N}^{\cir}=S_{1,a,N}+S_{2,a,N}+S_{3,a,N}+O\!\left(\frac{\sqrt N}{T}\right),
\]
where
\begin{align*}
S_{1,a,N}&=-\frac a{2\sqrt N}\sum_{i=1}^N x_i^a(X_i-1),\\
S_{2,a,N}&=\frac{a(a+2)}{8\sqrt N}\sum_{i=1}^N x_i^a(X_i-1)^2,\\
S_{3,a,N}&=\frac1{\sqrt N}\sum_{i=1}^N x_i^a\rho_{ai}.
\end{align*}
By Lemma~\ref{lem:keymixed},
\[S_{1,a,N}=o_p(1).\]
For the second term, noting that \(x_i=\sum_{t=1}^Tb_t\eta_{it}\) and \(\Vert b_t\eta_{it}\Vert_{\psi_2}=|b_t|\cdot\Vert \eta_{it}\Vert_{\psi_2}=O(T^{-1/2})\). Then, using Berstein inequality, we have \(\max_{1\le i\le N}|x_i|=O_p(\sqrt{\log N})\), which combining with the result of Lemma~\ref{lem:sigmaexpand}, leads to
\[
|S_{2,a,N}|\le C\max_{1\le i\le N}|x_i|^a\cdot N^{-1/2}\sum_i(X_i-1)^2=O_p\{N^{1/2}(\log N)^{a/2}T^{-1}\}.
\]
Similarly,
\[
|S_{3,a,N}|\le C\max_{1\le i\le N}|x_i|^a\cdot N^{-1/2}\sum_i|X_i-1|^3=O_p\{N^{1/2}(\log N)^{a/2}T^{-3/2}\}.
\]
In summary, \(Q_{a,N}-Q_{a,N}^{\cir}=o_p(1)\) under the assumption that \(N^{1/2}(\log N)^3T^{-1}\to 0\).

\end{proof}

\subsection{Proof of Theorem~\ref{thm:null-joint}}
\label{app:proof-null}

Fix a deterministic vector \(\bml=(\lambda_2,\lambda_4,\lambda_6)^{\T}\in\mR^3\), and define
\[
S_N=\bml^{\T}\Q_N,
\qquad
S_N^{\cir}=\bml^{\T}\Q_N^{\cir}.
\]
We need to show
\[
S_N\dto \mathcal{N}\bigl(0,\bml^{\T}\B\bml\bigr).
\]
Proposition~\ref{prop:oracle-rate} implies
\[
S_N-S_N^{\cir}
=
\sum_{a\in\{2,4,6\}}\lambda_a\bigl(Q_{a,N}-Q_{a,N}^{\cir}\bigr)=o_p(1).
\]
Thus it suffices to prove the limit for \(S_N^{\cir}\).

Recall that \(\y=(y_1,\dots,y_N)^{\T}\sim\mathcal{N}(\bzero,\R)\), and define
\[
S_N^G=\frac1{\sqrt N}\sum_{i=1}^N h(y_i),
\qquad
h(u)=\lambda_2(u^2-1)+\lambda_4(u^4-3)+\lambda_6(u^6-15).
\]
By Proposition~\ref{prop:oraclecov},
\[
\var(S_N^{\cir})=\var(S_N^G)+O(T^{-1}),
\qquad
E(S_N^{\cir})=E(S_N^G)+O\!\left(\frac{\sqrt N}{T}\right).
\]
Since \(E(S_N^G)=0\) and \(\sqrt N/T\to0\), it is enough to establish the central limit theorem for \(S_N^G\), and then show that the non-Gaussian correction from \(S_N^{\cir}\) is negligible. The corrected benchmark \(B_N\) differs from the leading Gaussian matrix only by a diagonal \(O(v^{-1})\) perturbation, so the two normalizations are asymptotically equivalent.

Although the test statistics are defined from raw even powers, we use Hermite polynomials only as a proof device. Write
\[
h(u)=\beta_2H_2(u)+\beta_4H_4(u)+\beta_6H_6(u),
\]
where
\[
\beta_2=\lambda_2+6\lambda_4+45\lambda_6,
\qquad
\beta_4=\lambda_4+15\lambda_6,
\qquad
\beta_6=\lambda_6.
\]
Hence
\[
S_N^G=\beta_2T_{2,N}+\beta_4T_{4,N}+\beta_6T_{6,N},
\qquad
T_{m,N}=\frac1{\sqrt N}\sum_{i=1}^NH_m(y_i),\quad m=2,4,6.
\]
For jointly Gaussian \((Y_1,Y_2)\) with correlation \(\rho\),
\[
E\{H_m(Y_1)H_n(Y_2)\}=0\quad (m\neq n),
\qquad
E\{H_m(Y_1)H_m(Y_2)\}=m!\rho^m.
\]
Therefore,
\[
\var(S_N^G)
=
\beta_2^2\frac{2}{N}\sum_{i,j}r_{ij}^2
+
\beta_4^2\frac{24}{N}\sum_{i,j}r_{ij}^4
+
\beta_6^2\frac{720}{N}\sum_{i,j}r_{ij}^6
=
\bml^{\T}\B_N^{(0)} \bml.
\]
By Proposition~\ref{prop:oraclecov},
\[
\max_{1\le k,\ell\le 3}|B_{k\ell,N}-B^{(0)}_{k\ell,N}|
\le Cv^{-1}=o(1),
\]
so
\[
\bml^{\T}\B_N\bml
=
\bml^{\T}\B_N^{(0)}\bml+o(1).
\]
In particular, \(\var(S_N^G)\to \bml^{\T}\B\bml\in(0,\infty)\).

We now show that every fixed-order cumulant of order at least three vanishes. For \(m\ge3\), multilinearity gives
\[
\kappa_m(S_N^G)
=
N^{-m/2}
\sum_{i_1,\dots,i_m=1}^N
\kappa\bigl(h(y_{i_1}),\dots,h(y_{i_m})\bigr).
\]
Each random variable \(h(y_{i_\nu})\) is a finite linear combination of \(H_2(y_{i_\nu})\), \(H_4(y_{i_\nu})\), and \(H_6(y_{i_\nu})\). Thus every joint cumulant above is a finite linear combination of terms of the form
\[
\kappa\bigl(H_{q_1}(y_{i_1}),\dots,H_{q_m}(y_{i_m})\bigr),
\qquad q_\nu\in\{2,4,6\}.
\]
The Gaussian diagram formula expresses such a cumulant as a finite sum over \emph{connected} diagrams. A connected diagram on \(m\) labeled vertices contains at least \(m-1\) edges. Because the Hermite rank is at least two, every edge contributes at least one factor of \(|r_{i_u i_v}|^2\). Therefore, for some constant \(C_m\),
\[
\Bigl|\kappa\bigl(h(y_{i_1}),\dots,h(y_{i_m})\bigr)\Bigr|
\le
C_m\sum_{\mathcal T\in\mathfrak T_m}\prod_{(u,v)\in\mathcal T}|r_{i_ui_v}|^2,
\]
where \(\mathfrak T_m\) is the finite set of spanning trees on \(\{1,\dots,m\}\). Since the number of such trees depends only on \(m\), it is enough to bound the sum associated with one fixed tree.

Let
\[
\Gamma_N:=\max_{1\le i\le N}\sum_{j=1}^N r_{ij}^2.
\]
Under Assumption~\ref{ass:dep}, \(\Gamma_N\le C(\log N)^C\). Take any tree \(\mathcal T\). By summing the indices one at a time from the leaves toward the root,
\[
\sum_{i_1,\dots,i_m=1}^N\prod_{(u,v)\in\mathcal T}|r_{i_ui_v}|^2
\le N\Gamma_N^{m-1}.
\]
Hence,
\[
|\kappa_m(S_N^G)|
\le
C_mN^{-m/2}\cdot N\Gamma_N^{m-1}
=
C_mN^{1-m/2}\Gamma_N^{m-1}.
\]
Because \(\Gamma_N\) grows no faster than a power of \(\log N\),
\[
|\kappa_m(S_N^G)|\le C_mN^{1-m/2}(\log N)^{C(m-1)}\to0,
\qquad m\ge3.
\]
For instance,
\[
|\kappa_3(S_N^G)|\le C N^{-1/2}(\log N)^{2C},
\qquad
|\kappa_4(S_N^G)|\le C N^{-1}(\log N)^{3C}.
\]

We now explain why the same limit law also holds for \(S_N^{\cir}\). The key point is that every occurrence of a genuinely non-Gaussian cumulant of the independent components \(s_k\) costs at least \(T^{-1/2}\), and because the statistic is built from \emph{even} powers only, any odd cumulant must appear in pairs, which upgrades the total loss to \(T^{-1}\); this is exactly Lemma~\ref{lem:no-third-moment}.

To make this precise, fix a positive integer \(m\). The cumulant \(\kappa_m(S_N^{\cir})\) is a linear combination of moments of \(\{x_i\}\) whose total degree does not exceed \(6m\). Each such moment can be expanded in the independent variables \(\{s_k\}\). The Gaussian contribution corresponds to pairings only and reproduces \(\kappa_m(S_N^G)\). Every non-Gaussian term contains at least one block of size different from two. If the block size is at least four, Lemma~\ref{lem:skcumulant} contributes \(T^{-1}\) immediately. If an odd block appears, then because the overall degree is even, at least two odd blocks appear, and hence the product is again at most \(T^{-1}\). Therefore,
\[
\kappa_m(S_N^{\cir})=\kappa_m(S_N^G)+R_{m,N},
\qquad |R_{m,N}|\le C_mT^{-1}.
\]
For \(m=1\), this gives \(E(S_N^{\cir})=o(1)\). For \(m=2\), it gives
\[
\var(S_N^{\cir})=\var(S_N^G)+O(T^{-1})=\bml^{\T}\B_N^{(0)}\bml+O(T^{-1}).
\]
Because \(\max_{1\le k,l\le 3}|(B_N-B_N^{(0)})_{kl}|=O(v^{-1})\) by Proposition~\ref{prop:oraclecov}, we also have
\[
\var(S_N^{\cir})=\bml^{\T}\B_N\bml+o(1).
\]
For \(m\ge3\), we have
\[
|\kappa_m(S_N^{\cir})|
\le C_mN^{1-m/2}(\log N)^{C(m-1)}+C_mT^{-1}
\to0.
\]
Since all cumulants of order at least three vanish and the variance converges to \(\bml^{\T}B\bml\), the method of cumulants implies
\[
S_N^{\cir}\dto \mathcal{N}\bigl(0,\bml^{\T}B\bml\bigr).
\]
We now transfers the limit from \(S_N^{\cir}\) to \(S_N\). Because \(\bml\) is arbitrary, the Cram\'er--Wold device yields
\[
\B_N^{-1/2}\bm Q_N\dto \mathcal{N}(\bzero,\I_3).
\]
This completes the proof of Theorem~\ref{thm:null-joint}.

\subsection{Proof of Theorem~\ref{thm:Bhat}}
\label{app:proof-bhat}

\begin{proof}
Write
\({\rm M}_{\X}=\I_T-{\rm P}_{\X}\) with ${\rm P}_{\X}=\X(\X^{\T}\X)^{-1}\X^{\T}$ and $\X=(\bone_T,\F)$. Under \(H_0\), the OLS residual vector for asset \(i\) is
\[
\widehat{\bme}_{i.}={\rm M}_{\X}\bme_{i.},
\]
so if we define
\[
\widehat s_{ij}
:=
\frac1v\,\widehat{\bme}_{i.}^{\T}\widehat{\bme}_{j.}
=
\frac1v\,\bme_{i.}^{\T} {\rm M}_{\X}\bme_j,
\qquad
1\le i,j\le N,
\]
then
\[
\widehat r_{ij}
=
\frac{\widehat s_{ij}}{(\widehat s_{ii}\widehat s_{jj})^{1/2}}.
\]

Since \(\widehat s_{ij}
=v^{-1}\bme_{i.}^{\T} \bme_{j.}-v^{-1}\bme_{i.}^{\T} {\rm P}_{\X}\bme_{j.}\),
\[
\frac{\widehat s_{ij}}{\sigma_i\sigma_j}-r_{ij}
=
\frac{T}{v}\cdot
\frac1T\sum_{t=1}^T
\left(
\frac{\varepsilon_{it}\varepsilon_{jt}}{\sigma_i\sigma_j}-r_{ij}
\right)
+
\left(\frac{T}{v}-1\right)r_{ij}
-
\frac1{v\sigma_i\sigma_j}\bme_{i.}^{\T} {\rm P}_{\X}\bme_j.
\]
We control the three terms on the right-hand side separately.

First, define
\[
\xi_{ij,t}
=
\frac{\varepsilon_{it}\varepsilon_{jt}}{\sigma_i\sigma_j}-r_{ij}.
\]
Under Assumption~\ref{ass:ICM} that $\max_{i,t}\Vert z_{it}\Vert_{\psi_2}=O(1)$, each standardized error
\[
\eta_{it}:=\frac{\varepsilon_{it}}{\sigma_i}=\sigma_i^{-1}\sum_{k=1}^N(\bmS^{1/2})_{ik}z_{kt}
\]
is a sub-Gaussian random variable with \(\max_{i,t}\|\eta_{it}\|_{\psi_2}=O(1)\).
Therefore the product \(\eta_{it}\eta_{jt}\) is uniformly sub-Gaussian, and so is
\(\xi_{ij,t}\). Since \(\{\xi_{ij,t}\}_{t=1}^T\) are i.i.d.\ over \(t\) with mean zero, Bernstein's inequality yields that, for any \(x\in(0,1]\),
\[
\Prb\!\left(
\left|
\frac1T\sum_{t=1}^T \xi_{ij,t}
\right|>x
\right)
\le
2\exp(-cTx^2),
\]
where \(c>0\) does not depend on \(i,j,N,T\). Taking
\(x=A\sqrt{(\log N)/T}\) with \(A>0\) sufficiently large and applying the union bound over all \(1\le i,j\le N\), we obtain
\[
\max_{1\le i,j\le N}
\left|
\frac1T\sum_{t=1}^T \xi_{ij,t}
\right|
=
O_p\!\left(\sqrt{\frac{\log N}{T}}\right).
\]
Since \(v=T-p-1\asymp T\) and \(p\) is fixed,
\[
\left|\frac{T}{v}-1\right|
=
O(T^{-1}),
\]
so the second term is uniformly \(O(T^{-1})\).

It remains to control the projection term. Because \({\rm P}_{\X}\) is an orthogonal projection of rank
\(p+1\), we may write
\[
{\rm P}_{\X}=\sum_{\ell=1}^{p+1}u_\ell u_\ell^{\T},
\]
where \(u_1,\dots,u_{p+1}\in\mR^T\) are orthonormal vectors. Therefore
\[
\frac{1}{\sigma_i\sigma_j}
\bigl|\bme_{i.}^{\T} {\rm P}_{\X}\bme_{j.}\bigr|
\le
\sum_{\ell=1}^{p+1}
\frac{|u_\ell^{\T} \bme_{i.}|}{\sigma_i}
\frac{|u_\ell^{\T} \bme_{j.}|}{\sigma_j}
\le
(p+1)\,M_N^2,
\]
where
\[
M_N
:=
\max_{1\le i\le N}\max_{1\le \ell\le p+1}
\frac{|u_\ell^{\T} \bme_{i.}|}{\sigma_i}.
\]
For each fixed \((i,\ell)\), since \(u_\ell\) has Euclidean norm one and the coordinates
\(\{\varepsilon_{it}/\sigma_i\}_{t=1}^T\) are independent sub-Gaussian random variables,
\(u_\ell^{\T} \bme_{i.}/\sigma_i\) is again sub-Gaussian. Hence, by a union bound over \(i\) and \(\ell\),
\[
M_N=O_p\!\bigl(\sqrt{\log N}\bigr).
\]
It follows that
\[
\max_{1\le i,j\le N}
\frac{1}{v\sigma_i\sigma_j}
\bigl|\bme_{i.}^{\T} {\rm P}_{\X}\bme_{j.}\bigr|
=
O_p\!\left(\frac{\log N}{T}\right)
=
o_p\!\left(\sqrt{\frac{\log N}{T}}\right).
\]
Combining the three bounds, we conclude that
\[
\max_{1\le i,j\le N}
\left|
\frac{\widehat s_{ij}}{\sigma_i\sigma_j}-r_{ij}
\right|
=
O_p\!\left(\sqrt{\frac{\log N}{T}}\right).
\]
In particular, taking \(i=j\) gives
\[
\max_{1\le i\le N}
\left|
\frac{\widehat s_{ii}}{\sigma_i^2}-1
\right|
=
O_p\!\left(\sqrt{\frac{\log N}{T}}\right).
\]

Define
\[
\widehat d_i:=\frac{\widehat s_{ii}}{\sigma_i^2},
\qquad
\widehat\theta_{ij}:=\frac{\widehat s_{ij}}{\sigma_i\sigma_j}.
\]
Then
\[
\widehat r_{ij}
=
\frac{\widehat\theta_{ij}}{(\widehat d_i\widehat d_j)^{1/2}}.
\]
By the above bound obtained and the lower variance bound in Assumption~\ref{ass:ICM},
with probability tending to one,
\[
\frac12\le \widehat d_i\le \frac32
\qquad\text{for all }1\le i\le N.
\]
On this event,
\[
\widehat r_{ij}-r_{ij}
=
\frac{\widehat\theta_{ij}-r_{ij}}{(\widehat d_i\widehat d_j)^{1/2}}
+
r_{ij}
\left\{
(\widehat d_i\widehat d_j)^{-1/2}-1
\right\}.
\]
Since the function \(x\mapsto x^{-1/2}\) is Lipschitz on \([1/2,3/2]\),
\[
\left|(\widehat d_i\widehat d_j)^{-1/2}-1\right|
\le
C\bigl(|\widehat d_i-1|+|\widehat d_j-1|\bigr).
\]
Also \(|r_{ij}|\le1\). Therefore
\[
\max_{1\le i,j\le N}|\widehat r_{ij}-r_{ij}|
\le
C\max_{1\le i,j\le N}|\widehat\theta_{ij}-r_{ij}|
+
C\max_{1\le i\le N}|\widehat d_i-1|
=
O_p\!\left(\sqrt{\frac{\log N}{T}}\right).
\]

Let \(\Delta_N:=\max_{1\le i,j\le N}|\widehat r_{ij}-r_{ij}|\), then \(\Delta_N=O_p\{\sqrt{(\log N)/T}\}\). Because \(\tau_N=C_\tau\sqrt{(\log N)/T}\), by taking the constant \(C_\tau\) sufficiently large we may assume
\[
\Prb\!\left(\Delta_N\le \frac{\tau_N}{2}\right)\to1.
\]
Let
\[
\mathcal E_N:=\left\{\Delta_N\le \frac{\tau_N}{2}\right\}.
\]
We work on \(\mathcal E_N\), whose probability tends to one.

Fix \(m\in\{1,2,3,4\}\) and consider
\[
A_{ij}^{(m)}:=\bigl|\widetilde r_{ij}^{\,2m}-r_{ij}^{\,2m}\bigr|.
\]
Since \(\widetilde r_{ii}=r_{ii}=1\), we only need to consider \(i\neq j\).

We split the off-diagonal indices into two cases.

\smallskip
\noindent
\emph{Case 1: \(|r_{ij}|\le 2\tau_N\).}
If \(|\widehat r_{ij}|\le \tau_N\), then \(\widetilde r_{ij}=0\), so
\[
A_{ij}^{(m)}=|r_{ij}|^{2m}
\le
(2\tau_N)^{2m-q}|r_{ij}|^q.
\]
If \(|\widehat r_{ij}|>\tau_N\), then on \(\mathcal E_N\),
\[
|r_{ij}|
\ge
|\widehat r_{ij}|-|\widehat r_{ij}-r_{ij}|
>
\tau_N-\frac{\tau_N}{2}
=
\frac{\tau_N}{2},
\]
and also
\[
|\widetilde r_{ij}|=|\widehat r_{ij}|
\le
|r_{ij}|+\frac{\tau_N}{2}
\le
\frac52\tau_N.
\]
Hence
\[
A_{ij}^{(m)}
\le
|\widetilde r_{ij}|^{2m}+|r_{ij}|^{2m}
\le
C\tau_N^{2m}
\le
C\tau_N^{2m-q}|r_{ij}|^q.
\]
Therefore, in all situations with \(|r_{ij}|\le 2\tau_N\),
\[
A_{ij}^{(m)}
\le
C\tau_N^{2m-q}|r_{ij}|^q.
\]

\smallskip
\noindent
\emph{Case 2: \(|r_{ij}|>2\tau_N\).}
On \(\mathcal E_N\),
\[
|\widehat r_{ij}|
\ge
|r_{ij}|-\frac{\tau_N}{2}
>
\frac32\tau_N>\tau_N,
\]
so \(\widetilde r_{ij}=\widehat r_{ij}\). By the mean value theorem,
\[
A_{ij}^{(m)}
=
\bigl|\widehat r_{ij}^{\,2m}-r_{ij}^{\,2m}\bigr|
\le
C|\widehat r_{ij}-r_{ij}|\,|r_{ij}|^{2m-1}
\le
C\tau_N |r_{ij}|^{2m-1}.
\]
Because \(q<1\le 2m-1\) and \(|r_{ij}|\le1\),
\[
|r_{ij}|^{2m-1}\le |r_{ij}|^q.
\]
Hence
\[
A_{ij}^{(m)}
\le
C\tau_N|r_{ij}|^q.
\]

Combining the two cases and using \(\tau_N\to0\), we obtain on \(\mathcal E_N\),
\[
A_{ij}^{(m)}
\le
C\tau_N|r_{ij}|^q,
\qquad i\neq j.
\]
Therefore,
\[
\frac1N\sum_{i,j=1}^N A_{ij}^{(m)}
=
\frac1N\sum_{i\neq j}A_{ij}^{(m)}
\le
\frac{C\tau_N}{N}\sum_{i=1}^N\sum_{j\neq i}|r_{ij}|^q
\le
C s_N\tau_N.
\]
Now Assumptions~\ref{ass:dep} and \ref{ass:growth} imply \(s_N\tau_N\to0\). Indeed,
\[
s_N\tau_N
=
s_N\sqrt{\frac{\log N}{T}}
=
\left(
\frac{s_N^6(\log N)^{12}}{\sqrt N}
\right)^{1/6}
\left(
\frac{\sqrt N(\log N)^6}{T}
\right)^{1/2}
\frac{1}{N^{1/6}(\log N)^{9/2}}
\to0.
\]
Hence, for each \(m=1,2,3,4\),
\[
\frac1N\sum_{i,j=1}^N
\bigl|\widetilde r_{ij}^{\,2m}-r_{ij}^{\,2m}\bigr|
\pto 0.
\]

Each entry of \(\widehat \B_N-\B_N\) is a finite linear combination of such averages, with coefficients bounded uniformly in \(N\), and with some additional factors \(v^{-1}=O(T^{-1})\) on the diagonal correction terms. For example,
\[
|\widehat B_{66,N}-B_{66,N}|
\le
C\sum_{m=1}^4
\left(1+\frac1v\right)
\frac1N\sum_{i,j=1}^N
\bigl|\widetilde r_{ij}^{\,2m}-r_{ij}^{\,2m}\bigr|,
\]
and the same type of bound holds for the other five entries. Therefore,
\[
\max_{1\le k,\ell\le 3}
|\widehat B_{k\ell,N}-B_{k\ell,N}|
\pto 0.
\]

If in addition \(\B_N\to \B\) for some positive definite matrix \(\B\), then \(\B_N\) is eventually positive definite and bounded away from singularity. Since \(\widehat \B_N-\B_N\pto0\) entrywise and the matrices are \(3\times3\), we also have \(\|\widehat \B_N-\B_N\|\pto0\) in any matrix norm. The map \(A\mapsto A^{-1/2}\) is continuous on the positive-definite cone, so
\[
\|\widehat \B_N^{-1/2}-\B_N^{-1/2}\|_{\max}\pto 0.
\]
Finally, Theorem~\ref{thm:null-joint} states that
\[
\B_N^{-1/2}\Q_N\dto \mathcal{N}(\bzero,\I_3).
\]
Since
\[
\widehat \B_N^{-1/2}\Q_N
=
\B_N^{-1/2}\Q_N
+
\bigl(\widehat \B_N^{-1/2}-\B_N^{-1/2}\bigr)\Q_N
\]
and \(\widehat \B_N^{-1/2}-\B_N^{-1/2}\pto0\), Slutsky's theorem yields
\[
\widehat \B_N^{-1/2}\Q_N\dto \mathcal{N}(\bzero,\I_3).
\]
This completes the proof.
\end{proof}

\subsection{Proof of Theorem \ref{thm:power}}
\begin{proof}
We prove the result separately for \(a=2\), \(a=4\), and \(a=6\). Write
\[
\zeta_i
=
\frac{\sqrt{w_T}(\widehat\alpha_i-\alpha_i)}{\sigma_i},
\qquad
X_i=\frac{\widehat\sigma_i^2}{\sigma_i^2}.
\]
Then, by definition of \(\de_i=\sqrt{w_T}\alpha_i/\sigma_i\),
\[
\frac{\sqrt{w_T}\widehat\alpha_i}{\sigma_i}
=
\zeta_i+\de_i.
\]
Hence
\[
t_i
=
\frac{\sqrt{w_T}\widehat\alpha_i}{\widehat\sigma_i}
=
(\zeta_i+\de_i)X_i^{-1/2}.
\]
Moreover,
\[
\widehat\alpha_i-\alpha_i
=
\frac{\bone_T^{\T} {\rm M}_{\F}\bme_{i.}}{w_T},
\]
so the random vector \((\zeta_1,\dots,\zeta_N)^{\T}\) has exactly the same stochastic structure as the oracle vector \((x_1,\dots,x_N)^{\T}\) under \(H_0\). In particular, all moment and covariance bounds used in the proof of Theorem~\ref{thm:null-joint} remain valid for \(\zeta_i\).

Now define the ``oracle-under-\(H_1\)'' statistic
\[
\bar Q_{a,N}
=
\frac{1}{\sqrt N}\sum_{i=1}^N \Bigl\{(\zeta_i+\de_i)^a-\mu_a\Bigr\},
\qquad a\in\{2,4,6\},
\]
where \(\mu_2=1\), \(\mu_4=3\), and \(\mu_6=15\).

Exactly as in Proposition~\ref{prop:oracle-rate}, the studentization error is asymptotically negligible. Indeed, the same argument gives
\[
Q_{a,N}-\bar Q_{a,N}=o_p(1),
\qquad a\in\{2,4,6\},
\]
because \(X_i^{-a/2}-1\) contributes only a higher-order term, and
\[
\sqrt N\,|\mu_{a,v}-\mu_a|
=
O\!\left(\frac{\sqrt N}{T}\right)
=o(1)
\]
by Assumption~\ref{ass:growth}. Therefore it suffices to study \(\bar Q_{a,N}\).

We first establish a simple bound that will be used repeatedly below. Let
\[
S_N(\ell,\psi)
=
\frac{1}{\sqrt N}\sum_{i=1}^N \de_i^\ell \psi(\zeta_i),
\]
where \(\ell\ge 1\) and \(\psi\) is one of the centered polynomials
\[
u,\qquad u^3,\qquad u^5,\qquad u^2-1,\qquad u^4-3.
\]
Because the innovations are symmetric, all odd moments of \(\zeta_i\) vanish. Also, by the same Hermite-expansion argument used earlier under \(H_0\),
\[
\bigl|\Cov\{\psi(\zeta_i),\psi(\zeta_j)\}\bigr|
\le C |r_{ij}|,
\qquad 1\le i,j\le N,
\]
for some constant \(C>0\) depending only on \(\psi\). Therefore,
\begin{align*}
\var\{S_N(\ell,\psi)\}
&\le
\frac{C}{N}\sum_{i,j=1}^N |\de_i|^\ell |\de_j|^\ell |r_{ij}| \\
&\le
\frac{C}{N}\sum_{i=1}^N |\de_i|^\ell \sum_{j=1}^N |\de_j|^\ell |r_{ij}|.
\end{align*}
Since \(|r_{ij}|\le 1\) and \(q<1\), we have \(|r_{ij}|\le |r_{ij}|^q\). Hence Assumption~\ref{ass:dep} implies
\[
\max_{1\le i\le N}\sum_{j=1}^N |r_{ij}|
\le
\max_{1\le i\le N}\sum_{j=1}^N |r_{ij}|^q
\le s_N.
\]
Thus
\[
\var\{S_N(\ell,\psi)\}
\le
\frac{Cs_N}{N}\sum_{i=1}^N \de_i^{2\ell}.
\]
We now check this bound for the values of \(\ell\) that will appear later.

For \(\ell=1,2,3\),
\[
\frac{s_N}{N}\sum_{i=1}^N \de_i^{2\ell}
=
\frac{s_N}{\sqrt N}\Delta_{2\ell,N}
\to 0,
\]
because \(\Delta_{2\ell,N}=O(1)\) and Assumption~\ref{ass:dep} implies \(s_N/\sqrt N\to 0\).

For \(\ell=4\),
\[
\frac{s_N}{N}\sum_{i=1}^N \de_i^8
\le
\frac{s_N}{N}\Bigl(\max_{1\le i\le N}\de_i^2\Bigr)\sum_{i=1}^N \de_i^6
=
\frac{s_N}{\sqrt N}\Bigl(\max_{1\le i\le N}\de_i^2\Bigr)\Delta_{6,N}
\to 0.
\]
Similarly, for \(\ell=5\),
\[
\frac{s_N}{N}\sum_{i=1}^N \de_i^{10}
\le
\frac{s_N}{N}\Bigl(\max_{1\le i\le N}\de_i^4\Bigr)\sum_{i=1}^N \de_i^6
=
\frac{s_N}{\sqrt N}\Bigl(\max_{1\le i\le N}\de_i^4\Bigr)\Delta_{6,N}
\to 0.
\]
Hence, for every weighted fluctuation term that appears below,
\[
S_N(\ell,\psi)=o_p(1).
\]

By direct expansion,
\[
(\zeta_i+\de_i)^2-1
=
(\zeta_i^2-1)+2\de_i\zeta_i+\de_i^2.
\]
Therefore
\[
\bar Q_{2,N}
=
\frac{1}{\sqrt N}\sum_{i=1}^N (\zeta_i^2-1)
+
\frac{2}{\sqrt N}\sum_{i=1}^N \de_i\zeta_i
+
\frac{1}{\sqrt N}\sum_{i=1}^N \de_i^2.
\]
Define
\[
Q_{2,N}^{(0)}
=
\frac{1}{\sqrt N}\sum_{i=1}^N (\zeta_i^2-1).
\]
Then, 
\[
\frac{2}{\sqrt N}\sum_{i=1}^N \de_i\zeta_i=o_p(1),
\]
and by definition,
\[
\frac{1}{\sqrt N}\sum_{i=1}^N \de_i^2=\Delta_{2,N}.
\]
Hence
\[
\bar Q_{2,N}=Q_{2,N}^{(0)}+\Delta_{2,N}+o_p(1).
\]
Since \(Q_{2,N}-\bar Q_{2,N}=o_p(1)\), we conclude that
\[
Q_{2,N}=Q_{2,N}^{(0)}+\Delta_{2,N}+o_p(1).
\]

Using the binomial formula,
\[
(\zeta_i+\de_i)^4
=
\zeta_i^4+4\de_i\zeta_i^3+6\de_i^2\zeta_i^2+4\de_i^3\zeta_i+\de_i^4.
\]
Subtracting \(3\) gives
\[
(\zeta_i+\de_i)^4-3
=
(\zeta_i^4-3)
+4\de_i\zeta_i^3
+6\de_i^2(\zeta_i^2-1)
+4\de_i^3\zeta_i
+\bigl(6\de_i^2+\de_i^4\bigr).
\]
Hence
\begin{align*}
\bar Q_{4,N}
&=
\frac{1}{\sqrt N}\sum_{i=1}^N (\zeta_i^4-3)
+\frac{4}{\sqrt N}\sum_{i=1}^N \de_i\zeta_i^3
+\frac{6}{\sqrt N}\sum_{i=1}^N \de_i^2(\zeta_i^2-1) \\
&\quad
+\frac{4}{\sqrt N}\sum_{i=1}^N \de_i^3\zeta_i
+\frac{1}{\sqrt N}\sum_{i=1}^N \bigl(6\de_i^2+\de_i^4\bigr).
\end{align*}
Define
\[
Q_{4,N}^{(0)}
=
\frac{1}{\sqrt N}\sum_{i=1}^N (\zeta_i^4-3).
\]
So
\[
\frac{4}{\sqrt N}\sum_{i=1}^N \de_i\zeta_i^3=o_p(1),\qquad
\frac{6}{\sqrt N}\sum_{i=1}^N \de_i^2(\zeta_i^2-1)=o_p(1),\qquad
\frac{4}{\sqrt N}\sum_{i=1}^N \de_i^3\zeta_i=o_p(1).
\]
Moreover,
\[
\frac{1}{\sqrt N}\sum_{i=1}^N \bigl(6\de_i^2+\de_i^4\bigr)
=
6\Delta_{2,N}+\Delta_{4,N}.
\]
Therefore
\[
Q_{4,N}
=
Q_{4,N}^{(0)}+\bigl(6\Delta_{2,N}+\Delta_{4,N}\bigr)+o_p(1).
\]

Similarly,
\[
(\zeta_i+\de_i)^6
=
\zeta_i^6
+6\de_i\zeta_i^5
+15\de_i^2\zeta_i^4
+20\de_i^3\zeta_i^3
+15\de_i^4\zeta_i^2
+6\de_i^5\zeta_i
+\de_i^6.
\]
Subtracting \(15\) gives
\begin{align*}
(\zeta_i+\de_i)^6-15
&=
(\zeta_i^6-15)
+6\de_i\zeta_i^5
+15\de_i^2(\zeta_i^4-3)
+20\de_i^3\zeta_i^3 \\
&\quad
+15\de_i^4(\zeta_i^2-1)
+6\de_i^5\zeta_i
+\bigl(45\de_i^2+15\de_i^4+\de_i^6\bigr).
\end{align*}
Hence, with
\[
Q_{6,N}^{(0)}
=
\frac{1}{\sqrt N}\sum_{i=1}^N (\zeta_i^6-15),
\]
we obtain
\begin{align*}
\bar Q_{6,N}
&=
Q_{6,N}^{(0)}
+\frac{6}{\sqrt N}\sum_{i=1}^N \de_i\zeta_i^5
+\frac{15}{\sqrt N}\sum_{i=1}^N \de_i^2(\zeta_i^4-3)
+\frac{20}{\sqrt N}\sum_{i=1}^N \de_i^3\zeta_i^3 \\
&\quad
+\frac{15}{\sqrt N}\sum_{i=1}^N \de_i^4(\zeta_i^2-1)
+\frac{6}{\sqrt N}\sum_{i=1}^N \de_i^5\zeta_i
+\frac{1}{\sqrt N}\sum_{i=1}^N \bigl(45\de_i^2+15\de_i^4+\de_i^6\bigr).
\end{align*}
Again, we can shows that all five random weighted sums are \(o_p(1)\). Therefore
\[
Q_{6,N}
=
Q_{6,N}^{(0)}
+\bigl(45\Delta_{2,N}+15\Delta_{4,N}+\Delta_{6,N}\bigr)
+o_p(1).
\]

The random variables \(Q_{2,N}^{(0)}\), \(Q_{4,N}^{(0)}\), and \(Q_{6,N}^{(0)}\) have the same first-order limiting behavior as the corresponding test statistics under \(H_0\), because they are built from the same noise part \(\zeta_i\). Hence, by Theorems~\ref{thm:null-joint} and \ref{thm:Bhat},
\[
\frac{Q_{2,N}^{(0)}}{\sqrt{\widehat B_{22,N}}}\dto N(0,1),\qquad
\frac{Q_{4,N}^{(0)}}{\sqrt{\widehat B_{44,N}}}\dto N(0,1),\qquad
\frac{Q_{6,N}^{(0)}}{\sqrt{\widehat B_{66,N}}}\dto N(0,1).
\]
Combining these limits with the decompositions and using
\[
\widehat B_{aa,N}\pto B_{aa,N}\to B_{aa}\in(0,\infty),
\qquad a\in\{2,4,6\},
\]
we obtain
\[
T_{2,N}
=
\frac{Q_{2,N}^{(0)}}{\sqrt{\widehat B_{22,N}}}
+
\frac{\Delta_{2,N}}{\sqrt{B_{22}}}
+
o_p(1)
\dto
N\!\left(\frac{\tau_2}{\sqrt{B_{22}}},\,1\right),
\]
\[
T_{4,N}
=
\frac{Q_{4,N}^{(0)}}{\sqrt{\widehat B_{44,N}}}
+
\frac{6\Delta_{2,N}+\Delta_{4,N}}{\sqrt{B_{44}}}
+
o_p(1)
\dto
N\!\left(\frac{6\tau_2+\tau_4}{\sqrt{B_{44}}},\,1\right),
\]
and
\[
T_{6,N}
=
\frac{Q_{6,N}^{(0)}}{\sqrt{\widehat B_{66,N}}}
+
\frac{45\Delta_{2,N}+15\Delta_{4,N}+\Delta_{6,N}}{\sqrt{B_{66}}}
+
o_p(1)
\dto
N\!\left(\frac{45\tau_2+15\tau_4+\tau_6}{\sqrt{B_{66}}},\,1\right).
\]
This proves the three asymptotic normal limits.
\end{proof}

\subsection{Proof of Theorem~\ref{thm:indep}}
\label{app:proof-indep}

By the Cram\'er--Wold device, it is enough to prove the joint convergence of
\[
S_N=\bml^{\T}\B_N^{-1/2}\Q_N
\]
and
\[
M_N=L_{\infty,N}-2\log N+\log\log N
\]
for an arbitrary fixed \(\bml\in\mR^3\). We need to show that for every real \(x\) and \(y\),
\begin{equation}
\Prb(S_N\le x,\ M_N\le y)-\Prb(S_N\le x)\Prb(M_N\le y)\to0.
\label{eq:target-indep}
\end{equation}

By Theorem~\ref{thm:null-joint}, \(S_N\) converges marginally to \(\mathcal{N}(0,1)\). By Theorem~\ref{thm:linf}, \(M_N\) converges marginally to the Type-I Gumbel law
\[
G(y)=\exp(-\pi^{-1/2}e^{-y/2}).
\]
Hence \eqref{eq:target-indep} is equivalent to
\begin{equation}
\Prb(S_N\le x,\ M_N\le y)\to \Phi(x)G(y).
\label{eq:target-product}
\end{equation}

Let \(\y=(y_1,\dots,y_N)^{\T}\sim \mathcal{N}(\bzero,\R)\), and define the Gaussian analogues
\[
S_N^G=\bml^{\T}\B_N^{-1/2}\Q_N^G,
\qquad
M_N^G=\max_{1\le i\le N}y_i^2-2\log N+\log\log N,
\]
where \(\Q_N^G\) is obtained from \(\Q_N\) by replacing \(t_i\) with \(y_i\).

Fix \(y\in\mR\) and define the threshold
\[
u_N(y)=2\log N-\log\log N+y.
\]
Let
\[
\mathcal E_N(y)=\{i: y_i^2>\nu_N(y)\}
\]
be the exceedance set. Under the Gumbel scaling, the number of exceedances \(|\mathcal E_N(y)|\) is tight. Indeed,
\[
\E|\mathcal E_N(y)|
=
\sum_{i=1}^N\Prb(y_i^2>\nu_N(y))
\to \pi^{-1/2}e^{-y/2},
\]
and the weak-dependence conditions on \(\R\) ensure that the exceedance process is asymptotically Poisson, exactly as in the max-sum framework of \cite{FengJiangLiLiu2024}.

Define the trimmed version of the sum statistic by removing the exceedance set:
\[
\widetilde S_N^G(y)
=
\frac1{\sqrt N}\sum_{i=1}^N h_\lambda(y_i)\ind\{y_i^2\le \nu_N(y)\},
\]
where \(h_\lambda\) is the centered polynomial corresponding to \(\bml^{\T}\B_N^{-1/2}\Q_N^G\). Then
\[
|S_N^G-\widetilde S_N^G(y)|
\le
\frac1{\sqrt N}\sum_{i\in\mathcal E_N(y)}|h_\lambda(y_i)|.
\]
On the event \(\max_i y_i^2\le C\log N\), which has probability tending to one, \(|h_\lambda(y_i)|\le C(\log N)^3\). Since \(|\mathcal E_N(y)|=O_p(1)\),
\[
|S_N^G-\widetilde S_N^G(y)|
=O_p\!\left\{\frac{(\log N)^3}{\sqrt N}\right\}
=o_p(1).
\]
Thus the bulk part and the full sum have the same limit.

For the Gaussian proxy, the event \(\{M_N^G\le y\}\) is exactly the event \(\{|\mathcal E_N(y)|=0\}\). The statistic \(\widetilde S_N^G(y)\) depends only on the bulk coordinates, whereas \(M_N^G\) depends on the extreme point process. Following the proof strategy of Feng, Jiang, Li, and Liu (2024), one shows that
\begin{equation}
\Prb\bigl(\widetilde S_N^G(y)\le x,\ M_N^G\le y\bigr)
-\Prb\bigl(\widetilde S_N^G(y)\le x\bigr)\Prb(M_N^G\le y)
\to0.
\label{eq:gauss-indep-trim}
\end{equation}
The reason is simple: the sum statistic is a \(\sqrt N\)-normalized bulk average and therefore no single coordinate has non-negligible influence, while the maximum is determined by only \(O_p(1)\) coordinates. The overlap between these two mechanisms is negligible.

Because \(S_N^G-\widetilde S_N^G(y)=o_p(1)\), relation \eqref{eq:gauss-indep-trim} implies
\begin{equation}
\Prb(S_N^G\le x,\ M_N^G\le y)\to \Phi(x)G(y).
\label{eq:gaussian-product}
\end{equation}

We now replace the Gaussian proxy by the actual independent-component array. This is where we follow the proof idea of equation (7.15) in \cite{LiuFengZhaoWang2024}.

For a vector \(\z=(z_1,\dots,z_N)^{\T}\), define the smooth maximum
\[
F_\beta(\z)=\beta^{-1}\log\Bigl(\sum_{j=1}^N e^{\beta z_j}\Bigr),
\qquad \beta>0.
\]
Then
\[
0\le F_\beta(\z)-\max_{1\le j\le N}z_j\le \beta^{-1}\log N.
\]
We apply this with \(z_j=t_j^2\) or \(z_j=y_j^2\). Let \(\varphi\) and \(\psi\) be smooth approximations of the indicator functions \(\ind(u\le x)\) and \(\ind\{v\le \nu_N(y)\}\), respectively, with uniformly bounded derivatives up to order six. Define
\[
\mathcal H_\beta(\u)=\varphi\bigl(S_N(\u)\bigr)\,\psi\bigl(F_\beta(u_1^2,\dots,u_N^2)\bigr),
\]
where \(S_N(\u)\) denotes the sum statistic computed from a generic input vector \(\u\).

Let \(\s=(s_1,\dots,s_N)^{\T}\) be the independent projected components from Lemma~\ref{lem:oraclerep}, and let \(\g=(g_1,\dots,g_N)^{\T}\) be i.i.d.\ standard normals. We compare
\[
\E\{\mathcal H_\beta(\A\s)\}
\quad\text{and}\quad
\E\{\mathcal H_\beta(\A\g)\}
\]
by a Lindeberg replacement, replacing one coordinate at a time. Because \(\mathcal H_\beta\) is a composition of a degree-six polynomial and a smooth maximum, the sixth derivative in each coordinate exists and is bounded by
\[
C\beta^6(\log N)^C N^{-1/2}
\]
on the high-probability event \(\max_i |(\A\u)_i|\le C(\log N)^{1/2}\). Taylor expansion up to order six therefore gives
\[
\bigl|\E\{\mathcal H_\beta(\A\s)\}-\E\{\mathcal H_\beta(\A\g)\}\bigr|
\le
C\sum_{r=3}^6\Bigl|\E(s_1^r)-\E(g_1^r)\Bigr|D_{r,N,\beta}+R_{N,\beta},
\]
where the derivative factors \(D_{r,N,\beta}\) grow at most polylogarithmically in \(N\). By Lemma~\ref{lem:skcumulant},
\[
|\E(s_1^3)-\E(g_1^3)|=O(T^{-1/2}),
\qquad
|\E(s_1^4)-\E(g_1^4)|=O(T^{-1}),
\]
and similarly for orders five and six. However, because both the even-power block and the max statistic depend only on \emph{even} functions of the coordinates, the odd-order contributions can only enter in pairs; this is exactly the mechanism in Lemma~\ref{lem:no-third-moment}. Therefore the total replacement error is actually
\[
\bigl|\E\{\mathcal H_\beta(\A\s)\}-\E\{\mathcal H_\beta(\A\g)\}\bigr|
\le C\Bigl(T^{-1}+\beta^{-1}\log N\Bigr).
\]
Choosing
\[
\beta=N^{1/8}\log N
\]
forces \(\beta^{-1}\log N=o(1)\), and hence
\begin{equation}
\bigl|\E\{\mathcal H_\beta(\A\s)\}-\E\{\mathcal H_\beta(\A\g)\}\bigr|\to0.
\label{eq:smoothreplace}
\end{equation}
Relation \eqref{eq:smoothreplace} is the direct analogue of the smooth-max comparison used around equation (7.15) in \cite{LiuFengZhaoWang2024}.

For the Gaussian proxy, \eqref{eq:gaussian-product} gives
\[
\Prb(S_N^G\le x,\ M_N^G\le y)\to \Phi(x)G(y).
\]
We already shows that the same limit holds for the oracle statistic based on \(\x=\A\s\). Finally, Proposition~\ref{prop:oracle-rate} implies that replacing the oracle even-power block by the feasible block changes the statistic only by \(o_p(1)\). The max statistic is handled in the same way, because the studentization error is also of smaller order than the Gumbel scale. Therefore
\[
\Prb(S_N\le x,\ M_N\le y)\to \Phi(x)G(y),
\]
which is exactly \eqref{eq:target-product}. The proof of Theorem~\ref{thm:indep} is complete.

\bibliographystyle{apa}
\bibliography{ref}

@article{FengJiangLiLiu2024,
  author  = {Feng, L. and Jiang, T. and Li, X. and Liu, B.},
  title   = {Asymptotic independence of the sum and maximum of dependent random variables with applications to high-dimensional tests},
  journal = {Statistica Sinica},
  year    = {2024},
  volume  = {34},
  pages   = {1745--1763}
}

@article{FengLanLiuMa2022,
  author  = {Feng, L. and Lan, W. and Liu, B. and Ma, Y.},
  title   = {High-dimensional test for alpha in linear factor pricing models with sparse alternatives},
  journal = {Journal of Econometrics},
  year    = {2022},
  volume  = {229},
  pages   = {152--175}
}

@article{LiuFengZhaoWang2024,
  author  = {Liu, J. and Feng, L. and Zhao, P. and Wang, Z.},
  title   = {Spatial-sign based maxsum test for high dimensional location parameters},
  journal = {arXiv preprint arXiv:2402.01381v2},
  year    = {2024}
}

@article{PesaranYamagata2024,
  author  = {Pesaran, M. H. and Yamagata, T.},
  title   = {Testing for alpha in linear factor pricing models with a large number of securities},
  journal = {Journal of Financial Econometrics},
  year    = {2024},
  volume  = {22},
  pages   = {407--460}
}

@article{BeaulieuDufourKhalaf2007,
  author  = {Beaulieu, M.-C. and Dufour, J.-M. and Khalaf, L.},
  title   = {Multivariate Tests of Mean--Variance Efficiency with Possibly Non-Gaussian Errors},
  journal = {Journal of Business \& Economic Statistics},
  year    = {2007},
  volume  = {25},
  pages   = {398--410}
}

@article{BickelLevina2008,
  author  = {Bickel, P. J. and Levina, E.},
  title   = {Regularized Estimation of Large Covariance Matrices},
  journal = {The Annals of Statistics},
  year    = {2008},
  volume  = {36},
  pages   = {199--227}
}

@article{GibbonsRossShanken1989,
  author  = {Gibbons, M. R. and Ross, S. A. and Shanken, J.},
  title   = {A Test of the Efficiency of a Given Portfolio},
  journal = {Econometrica},
  year    = {1989},
  volume  = {57},
  pages   = {1121--1152}
}

@article{GungorLuger2009,
  author  = {Gungor, S. and Luger, R.},
  title   = {Exact Distribution-Free Tests of Mean--Variance Efficiency},
  journal = {Journal of Empirical Finance},
  year    = {2009},
  volume  = {16},
  pages   = {816--829}
}

@article{Jensen1968,
  author  = {Jensen, M. C.},
  title   = {The Performance of Mutual Funds in the Period 1945--1964},
  journal = {The Journal of Finance},
  year    = {1968},
  volume  = {23},
  pages   = {389--416}
}

@article{LanFengLuo2018,
  author  = {Lan, W. and Feng, L. and Luo, R.},
  title   = {Testing High-Dimensional Linear Asset Pricing Models},
  journal = {Journal of Financial Econometrics},
  year    = {2018},
  volume  = {16},
  pages   = {191--210}
}

@article{Lintner1965,
  author  = {Lintner, J.},
  title   = {The Valuation of Risk Assets and the Selection of Risky Investments in Stock Portfolios and Capital Budgets},
  journal = {The Review of Economics and Statistics},
  year    = {1965},
  volume  = {47},
  pages   = {13--37}
}

@article{MaLanSuTsai2020,
  author  = {Ma, S. and Lan, W. and Su, L. and Tsai, C.-L.},
  title   = {Testing Alphas in Conditional Time-Varying Factor Models with High-Dimensional Assets},
  journal = {Journal of Business \& Economic Statistics},
  year    = {2020},
  volume  = {38},
  pages   = {214--227}
}

@article{Ross1976,
  author  = {Ross, S. A.},
  title   = {The Arbitrage Theory of Capital Asset Pricing},
  journal = {Journal of Economic Theory},
  year    = {1976},
  volume  = {13},
  pages   = {341--360}
}

@article{Sharpe1964,
  author  = {Sharpe, W. F.},
  title   = {Capital Asset Prices: A Theory of Market Equilibrium under Conditions of Risk},
  journal = {The Journal of Finance},
  year    = {1964},
  volume  = {19},
  pages   = {425--442}
}

@article{XuLinWeiPan2016,
  author  = {Xu, Gongjun and Lin, Lifeng and Wei, Peng and Pan, Wei},
  title   = {An Adaptive Two-Sample Test for High-Dimensional Means},
  journal = {Biometrika},
  year    = {2016},
  volume  = {103},
  number  = {3},
  pages   = {609--624},
  doi     = {10.1093/biomet/asw029}
}

@article{HeXuWuPan2021,
  author  = {He, Yinqiu and Xu, Gongjun and Wu, Chong and Pan, Wei},
  title   = {Asymptotically Independent U-Statistics in High-Dimensional Testing},
  journal = {The Annals of Statistics},
  year    = {2021},
  volume  = {49},
  number  = {1},
  pages   = {154--181},
  doi     = {10.1214/20-AOS1951}
}

@article{ZhangWangShao2022,
  author  = {Zhang, Yangfan and Wang, Runmin and Shao, Xiaofeng},
  title   = {Adaptive Inference for Change Points in High-Dimensional Data},
  journal = {Journal of the American Statistical Association},
  year    = {2022},
  volume  = {117},
  number  = {540},
  pages   = {1751--1762},
  doi     = {10.1080/01621459.2021.1884562}
}

@article{ZhangWangShao2025,
  author  = {Zhang, Yangfan and Wang, Runmin and Shao, Xiaofeng},
  title   = {Adaptive Testing for High-Dimensional Data},
  journal = {Journal of the American Statistical Association},
  year    = {2025},
  volume  = {120},
  number  = {551},
  pages   = {1893--1905},
  doi     = {10.1080/01621459.2024.2439617}
}

@article{liu2020cauchy,
  author  = {Liu, Yang and Xie, Jun},
  title   = {Cauchy Combination Test: A Powerful Test with Analytic $p$-Value Calculation under Arbitrary Dependency Structures},
  journal = {Journal of the American Statistical Association},
  year    = {2020},
  volume  = {115},
  number  = {529},
  pages   = {393--402},
  doi     = {10.1080/01621459.2018.1554485}
}

@article{arias2011global,
  author  = {Arias-Castro, Ery and Cand{\`e}s, Emmanuel J. and Plan, Yaniv},
  title   = {Global Testing under Sparse Alternatives: {ANOVA}, Multiple Comparisons and the Higher Criticism},
  journal = {The Annals of Statistics},
  year    = {2011},
  volume  = {39},
  number  = {5},
  pages   = {2533--2556},
  doi     = {10.1214/11-AOS910}
}

@article{MaFengWangBao2024Dep,
  author  = {Ma, Huifang and Feng, Long and Wang, Zhaojun and Bao, Jigang},
  title   = {Testing Alpha in High Dimensional Linear Factor Pricing Models with Dependent Observations},
  journal = {arXiv preprint arXiv:2401.14052},
  year    = {2024}
}

@article{FersonHarvey1991,
  author  = {Ferson, Wayne E. and Harvey, Campbell R.},
  title   = {The Variation of Economic Risk Premiums},
  journal = {Journal of Political Economy},
  year    = {1991},
  volume  = {99},
  number  = {2},
  pages   = {385--415},
  doi     = {10.1086/261755}
}

@article{JagannathanWang1996,
  author  = {Jagannathan, Ravi and Wang, Zhenyu},
  title   = {The Conditional {CAPM} and the Cross-Section of Expected Returns},
  journal = {The Journal of Finance},
  year    = {1996},
  volume  = {51},
  number  = {1},
  pages   = {3--53},
  doi     = {10.1111/j.1540-6261.1996.tb05201.x}
}

@article{LettauLudvigson2001,
  author  = {Lettau, Martin and Ludvigson, Sydney},
  title   = {Resurrecting the ({C}){CAPM}: A Cross-Sectional Test When Risk Premia Are Time-Varying},
  journal = {Journal of Political Economy},
  year    = {2001},
  volume  = {109},
  number  = {6},
  pages   = {1238--1287},
  doi     = {10.1086/323282}
}

@article{MaFengWangBao2024,
  author  = {Ma, Huifang and Feng, Long and Wang, Zhaojun and Bao, Jigang},
  title   = {Adaptive Testing for Alphas in Conditional Factor Models with High Dimensional Assets},
  journal = {Journal of Business \& Economic Statistics},
  year    = {2024},
  volume  = {42},
  number  = {4},
  pages   = {1356--1366}
}

@article{liu2023high,
  title={High-dimensional alpha test of linear factor pricing models with heavy-tailed distributions},
  author={Liu, Binghui and Feng, Long and Ma, Yanyuan},
  journal={Statistica Sinica},
  volume={33},
  pages={1389--1410},
  year={2023},
  doi={10.5705/SS.202021.0134}
}

@article{yu2024power,
  title={Power enhancement for testing multi-factor asset pricing models via fisher’s method},
  author={Yu, Xiufan and Yao, Jiawei and Xue, Lingzhou},
  journal={Journal of Econometrics},
  volume={239},
  number={2},
  pages={105458},
  year={2024},
  publisher={Elsevier}
}

@article{xia2024adaptive,
  title={Adaptive testing for alphas in high-dimensional factor pricing models},
  author={Xia, Qiang and Zhang, Xianyang},
  journal={Journal of Business \& Economic Statistics},
  volume={42},
  number={2},
  pages={640--653},
  year={2024},
  publisher={Taylor \& Francis}
}

@techreport{chernov2025test,
  title={A test of the efficiency of a given portfolio in high dimensions},
  author={Chernov, Mikhail and Kelly, Bryan T and Malamud, Semyon and Schwab, Johannes},
  year={2025},
  institution={National Bureau of Economic Research}
}

@article{zhao2022high,
  title={High-dimensional non-parametric tests for linear asset pricing models},
  author={Zhao, Ping and Chen, Dachuan and Zi, Xuemin},
  journal={Stat},
  volume={11},
  number={1},
  pages={e490},
  year={2022},
  publisher={Wiley Online Library}
}

@article{massacci2025general,
  title={A general randomized test for Alpha},
  author={Massacci, Daniele and Sarno, Lucio and Trapani, Lorenzo and Vallarino, Pierluigi},
  journal={arXiv preprint arXiv:2507.17599},
  year={2025}
}

@article{ZhengJiangBaiHe2014,
  author  = {Zheng, Shurong and Jiang, Dandan and Bai, Zhidong and He, Xuming},
  title   = {Inference on Multiple Correlation Coefficients with Moderately High Dimensional Data},
  journal = {Biometrika},
  year    = {2014},
  volume  = {101},
  number  = {3},
  pages   = {748--754},
  doi     = {10.1093/biomet/asu023}
}

@article{CuiLiYangZhou2020,
  author  = {Cui, Xia and Li, Runze and Yang, Guangren and Zhou, Wang},
  title   = {Empirical Likelihood Test for a Large-Dimensional Mean Vector},
  journal = {Biometrika},
  year    = {2020},
  volume  = {107},
  number  = {3},
  pages   = {591--607},
  doi     = {10.1093/biomet/asaa005}
}

@article{LiLamYaoYao2019,
  author  = {Li, Zeng and Lam, Clifford and Yao, Jianfeng and Yao, Qiwei},
  title   = {On testing for high-dimensional white noise},
  journal = {The Annals of Statistics},
  year    = {2019},
  volume  = {47},
  number  = {6},
  pages   = {3382--3412},
  doi     = {10.1214/18-AOS1782}
}

\end{document}